\documentclass[11pt]{article}
\usepackage[margin=1in]{geometry}

\usepackage{microtype}
\usepackage{soul}

\usepackage[hyphens]{url} 

\makeatletter
\newcommand\email[2][]%
   {\newaffiltrue\let\AB@blk@and\AB@pand
      \if\relax#1\relax\def\AB@note{\AB@thenote}\else\def\AB@note{\relax}%
        \setcounter{Maxaffil}{0}\fi
      \begingroup
        \let\protect\@unexpandable@protect
        \def\thanks{\protect\thanks}\def\footnote{\protect\footnote}%
        \@temptokena=\expandafter{\AB@authors}%
        {\def\\{\protect\\\protect\Affilfont}\xdef\AB@temp{#2}}%
         \xdef\AB@authors{\the\@temptokena\AB@las\AB@au@str
         \protect\\[\affilsep]\protect\Affilfont\AB@temp}%
         \gdef\AB@las{}\gdef\AB@au@str{}%
        {\def\\{, \ignorespaces}\xdef\AB@temp{#2}}%
        \@temptokena=\expandafter{\AB@affillist}%
        \xdef\AB@affillist{\the\@temptokena \AB@affilsep
          \AB@affilnote{}\protect\Affilfont\AB@temp}%
      \endgroup
       \let\AB@affilsep\AB@affilsepx
}
\makeatother

\usepackage[noblocks]{authblk}
\renewcommand\Affilfont{\normalsize}
\usepackage[utf8]{inputenc}
\usepackage[small]{caption}
\usepackage{graphicx}
\usepackage{float}
\usepackage{booktabs}
\usepackage[switch]{lineno}
\usepackage{amsfonts}
\usepackage{amsthm}
\usepackage{amssymb}
\usepackage{amsmath}
\usepackage{times}
\usepackage[ruled,noend]{algorithm2e}
\usepackage{comment}
\LinesNumbered 
\usepackage[dvipsnames]{xcolor}
\usepackage{setspace}
\usepackage{thmtools}
\usepackage{mathtools}
\usepackage{enumitem}
\usepackage{xparse}
\usepackage{verbatim}
\usepackage{natbib}
\usepackage{import}
\usepackage[colorlinks, citecolor=NavyBlue, linkcolor=teal]{hyperref}
\hypersetup{breaklinks=true}
\usepackage[nameinlink]{cleveref}
\newcommand{\supp}{\mathrm{supp}}

\DeclareMathOperator*{\EX}{\mathbb{E}}

\usepackage[nameinlink]{cleveref}

\usepackage{thm-restate}

\newtheorem{theorem}{Theorem}[section]
\newtheorem{lemma}[theorem]{Lemma}

\newtheorem{corollary}[theorem]{Corollary}

\newtheorem{definition}[theorem]{Definition}

\DeclareMathOperator*{\argmax}{arg\,max}
\renewcommand{\vec}[1]{\mathbf{#1}}

\makeatletter
\NewDocumentEnvironment{algoindent}{m}{%
  \algocf@push{#1}\hbox\bgroup\vtop\bgroup
  \algocf@addskiptotal
}{%
  \egroup\egroup
  \algocf@pop{#1}%
}
\newenvironment{proofof}[1]
{\smallskip\noindent{\emph{#1.}}}
{\hfill$\Box$\medskip}

\makeatother
\NewDocumentEnvironment{algolabel}{m o}{%
  \setbox0\hbox{\underline{#1} : }
  \dimen0=\wd0
  \begingroup
    \leavevmode\box0 \IfValueT{#2}{#2\\}%
  \endgroup
  \expandafter\algoindent\expandafter{\the\dimen0}
}{\endalgoindent}

\newcommand{\crefdeffpart}[2]{%
  \hyperref[#2]{\namecref{#1}~\labelcref*{#1}\Cref*{#2}}%
}
\newcommand{\Crefdeffpart}[2]{%
  \hyperref[#2]{\nameCref{#1}~\labelcref*{#1}\Cref*{#2}}%
}

\let\oldnl\nl
\newcommand{\nonl}{\renewcommand{\nl}{\let\nl\oldnl}}

\makeatletter
\newenvironment{mechanism}[1][htb]{%
    
   \begin{algorithm}[#1]%
  }{\end{algorithm}}
\makeatother
\Crefname{algorithm}{Mechanism}{Mechanisms}

\usepackage{tikz}

\usepackage{color}

\DeclarePairedDelimiter{\ceil}{\lceil}{\rceil}

\DeclarePairedDelimiter{\floor}{\lfloor}{\rfloor}

\renewcommand{\vec}[1]{\mathbf{#1}}

\allowdisplaybreaks

\begin{document}

\title{Online Budget-Feasible Mechanism Design with Predictions}

\author[1,3]{Georgios Amanatidis}
\author[1,3,5]{Evangelos Markakis}
\author[1,3]{Christodoulos Santorinaios}
\author[2,4]{\\Guido Sch\"afer}
\author[1]{Panagiotis Tsamopoulos}
\author[6]{Artem Tsikiridis}

\affil[1]{Athens University of Economics and Business, Greece}
\affil[2]{Centrum Wiskunde \& Informatica (CWI), The Netherlands}
\affil[3]{Archimedes, Athena Research Center, Greece}
\affil[4]{University of Amsterdam, The Netherlands}
\affil[5]{Input Output Global (IOG), Greece}
\affil[6]{Technical University of Munich, Germany}

\maketitle

\begin{abstract}
\noindent Augmenting the input of algorithms with predictions is an algorithm design paradigm that suggests leveraging a (possibly erroneous) prediction to improve worst-case performance guarantees when the prediction is perfect (consistency), while also providing a performance guarantee when the prediction fails (robustness). Recently, \citet{xu2022mechanism} and \citet{agrawal2022learning} proposed to consider settings with strategic agents under this framework. 
   In this paper, we initiate the study of budget-feasible mechanism design with predictions. These mechanisms model a procurement auction scenario in which an auctioneer (buyer) with a strict budget constraint seeks to purchase goods or services from a set of strategic agents, so as to maximize her own valuation function. We focus on the \textit{online} version of the problem where the arrival order of agents is random. We design mechanisms that are truthful, budget-feasible, and achieve a significantly improved competitive ratio for both monotone and non-monotone submodular valuation functions compared to their state-of-the-art counterparts without predictions. Our results assume access to a prediction for the value of the optimal solution to the offline problem. 
   We complement our positive results by showing that for the offline version of the problem, access to predictions is mostly ineffective in improving approximation guarantees. 
\end{abstract}

\thispagestyle{empty}

\section{Introduction}
Our work revolves around the design of auction mechanisms for procurement auctions. Procurement auctions, also known as reverse auctions, are mechanisms for assigning a set of projects or tasks to a set of candidate providers. In the more traditional applications of reverse auctions, the auctioneer may be a governmental organization (e.g., for assigning a public project to a construction firm), or a company interested in outsourcing tasks, such as the supply of logistics or retail services; see \citet{CSS06}. In  recent years, the explosion and monetization of the Internet have created even further  applications and platforms where even a single individual can conduct online assignments/auctions. Two such prominent families of examples include platforms for online labour markets, such as Upwork or Freelancer, and platforms for crowdsourcing systems, where the posted tasks can range from very elementary assignments as image labeling, all the way to more specialized services like translating documents. Participatory crowdsensing also falls under this framework, where some entities are interested in collecting sensor information from smartphones or other devices. 

The emergence of such online markets has caught the increasing attention of the research community, which has focused on properly formulating and addressing relevant research questions. One of the popular models in the literature for capturing procurement auctions is that of budget-feasible mechanism design, originally proposed in the seminal paper of \citet{singer2010budget}. 
This model concerns a scenario where an auctioneer with a strict budget constraint seeks to purchase goods or services from a set of strategic agents, who may misreport their cost to their advantage for obtaining higher payments. Under these considerations, a natural goal for the auctioneer is to come up with a truthful mechanism for hiring a subset of the agents that maximizes the procured value (measured by the valuation function of the auctioneer), and such that the total payments to the agents respect the budget limitation. 
Given that even the non-strategic versions of such budget-constrained problems tend to be NP-hard, the main focus is on providing budget-feasible mechanisms that achieve approximation guarantees on the auctioneer's optimal potential value. This has led to a steady stream of works which we outline in our related work section.

The focus of our work is to study budget-feasible mechanisms under the rather recently introduced framework of learning-augmented mechanism design. This forms a new paradigm in designing algorithms, where the input is augmented by a prediction on some relevant parameter of the problem, possibly obtained by past data. The main goal is to leverage a (possibly erroneous) prediction to improve worst-case performance guarantees when the prediction is accurate (referred to as consistency), while also providing a performance guarantee when it is not (referred to as robustness),  without knowing  the prediction accuracy in advance. Recently, 
\citet{xu2022mechanism} and \citet{agrawal2022learning}
proposed to consider settings with strategic agents through this lens. Motivated by these two initial attempts, there has been a fast growing interest and a series of other works that studied the design of truthful mechanisms with predictions in a variety of settings, leading mostly to positive results (i.e.,  improved approximation guarantees under correct predictions while not far off from the best known guarantee otherwise).

\paragraph{Our Contribution}
We initiate the study of budget-feasible mechanism design with predictions. 
We focus on the online version of the problem in the random order model, where the agents arrive in a uniformly random order, and under a  submodular valuation function for the auctioneer. 
Further, the input of our mechanisms is augmented by a prediction for the value of the optimal solution to the offline problem. Our main results can be summarized as follows:

\begin{itemize}[itemsep=1pt, topsep=1pt, leftmargin=10pt]

    \item In \Cref{sec:mono-submodular} we propose a family of universally truthful, budget-feasible mechanisms (\Cref{mechanism:mono-sm-conv}; parameterized by an error tolerance parameter $\tau$) for monotone submodular objectives which attain $O(1)$ consistency and robustness. For $\tau$ close to $1$, the consistency, i.e., the approximation guarantee when the prediction is perfect, can be as low as $6$, which is a significant improvement even versus the best known approximation guarantee of $54.4$ for Submodular Knapsack Secretary, i.e., the non-strategic version of the problem, without predictions (due to \citet{feldman11}). 
    
    \item For smaller values of $\tau$, the robustness of \Cref{mechanism:mono-sm-conv} in \Cref{sec:mono-submodular}, i.e., the approximation guarantee when the prediction is arbitrarily bad, can be as low as $146$. This also implies a $146$-approximation mechanism \textit{without predictions}, improving the previously known approximation of $1710$ by \citet{amanatidis2022budget} by a factor greater than $11$!

    \item In \Cref{sec:mono-sm-sampling}, we present a different mechanism (\Cref{mechanism:mono-sm2}) 
    where the prediction implicitly determines the length of the initial agent-sampling window (although here we have fixed the latter for the sake of presentation). Compared to \Cref{mechanism:mono-sm-conv}, this approach is more conservative, as \Cref{mechanism:mono-sm2} is, in a sense, less sensitive to (very) erroneous predictions. In particular, we show that our mechanism achieves a consistency-robustness tradeoff of~$95$ and~$280$.

    \item We obtain analogous results for non-monotone submodular objectives, albeit with some loss on approximation, via \Cref{mechanism:nonmono-sm-conv} and \Cref{mechanism:non-mono-sm2}. Specifically, for $\tau$ close to $1$, the consistency of \Cref{mechanism:nonmono-sm-conv} can be as low as $19$, whereas, for smaller values of $\tau$, its robustness can be as low as $445$. We emphasize once more that, even when one completely ignores the prediction, our analysis leads to a large improvement on the current state-of-the-art approximation due to \citet{amanatidis2022budget}, which is $1710$ even for Submodular Knapsack Secretary (the non-strategic version of the problem). Finally, we show that the result of the analogue of \Cref{mechanism:mono-sm2} for non-monotone submodular valuations (\Cref{mechanism:non-mono-sm2}) achieves a consistency-robustness tradeoff of  $228$ and $818$.

    \item We complement our positive results by showing in \Cref{sec:lower-bound} that for the \textit{offline} version, access to predictions is mostly ineffective in improving approximation guarantees. In particular, we show that no randomized, universally truthful, budget-feasible mechanism with bounded robustness can achieve a consistency guarantee better than the lower bound of $2$ that is known to hold for additive valuations and randomized mechanisms \citep{chen2011approximability}. 
\end{itemize}

All our mechanisms run in polynomial time assuming access to a value oracle for the objective function, which is standard in submodular optimization. From a technical point of view, we design posted price auction mechanisms, the prices of which are determined by exploiting the information from the
prediction on the optimal value, the outcome of an initial
sampling phase, or a combination of the two, depending on the mechanism. Overall, our results shed light on the power but also on the limitations of the learning-augmented paradigm for budget-feasible mechanisms.

\paragraph{Related Work}
The design of truthful budget-feasible mechanisms was initiated by \citet{singer2010budget} for additive and monotone submodular valuation functions, and has sparked a rich line of work. For additive valuation functions a best-possible $2$-approximation randomized mechanism was given by \citet{gravin2020optimal}, whereas it is also known that the best possible factor of any deterministic mechanism is between $1+\sqrt{2}$ and $3$ (due to \citet{chen2011approximability} and \citet{gravin2020optimal}). For monotone submodular valuation functions, through a series of improvements  (see e.g., \citet{chen2011approximability,jalaly2021simple,balkanski2022deterministic}), the current state of the art is $4.45$ for deterministic  and $4.08$ for randomized mechanisms due to \citet{han2024triple} who follow the clock auction paradigm. 
Without the monotonicity assumption, obtaining any $O(1)$ approximation seemed to be more challenging (see e.g., \citet{amanatidis2022budget}), but here as well the recent improvements of \citet{balkanski2022deterministic,huang2023randomized} have led to a relatively small approximation ratio, namely $11.67$, again by  \citet{han2024triple}.

There have also been works focusing on richer valuation classes like XOS and subadditive functions (see e.g. \citet{bei2017worst,amanatidis2017budget,neogi_et_al:LIPIcs.ITCS.2024.84,loglogn}), exploring beyond-the-worst-case scenarios (\citet{rubinstein23}), or using large market assumptions, see \citet{AnariGN18,jalaly2021simple}. Very recently the problem has been considered in a multi-parameter domain by \citet{neogi25multidimensional}.
We refer the reader to the recent survey of \citet{liu24} for an overview of various settings (e.g., environments with combinatorial constraints and more general single-parameter domains).

A recent direction in theoretical computer science is the area of ``beyond the worst-case analysis'' (see, e.g., \citet{roughgarden2021beyond}). One prominent approach within this framework is \textit{algorithms with predictions}, which originated in the field of online algorithms (see, e.g., \citet{lykouris2021competitive}) and utilizes predictions, e.g.,  provided by an ML algorithm trained on historical data, to overcome worst case lower bounds.
The idea has been also applied in strategic settings, where predictions are used to simulate (an aggregation of) the agents' private information, starting with  the work of \citet{agrawal2022learning} and \citet{xu2022mechanism}. Areas of application in mechanism design include auction-related environments (see e.g.,  \citet{prasad23,balkanski23,christodoulou24,caragiannis24,balkanski23scheduling,gkatzelis25}), online mechanisms (see e.g., \citet{balkanski23,lu24}), mechanisms without money (see e.g. \citet{balkanski24randomized,colini24,barak24,cohen24}) and computational social choice (see e.g. \citet{filos2025}). For an almost exhaustive list of papers in the area, see the online repository of  \citet{algorithms_with_predictions}.%

The \textit{online} budget-feasible mechanism design problem we study was introduced by \citet{badanidiyuru2012learning} and is a generalization to the strategic setting of the Submodular Knapsack Secretary problem (see, e.g., \citet{babaioff17}). The latter has also been considered by \citet{bateni13} and subsequently by \citet{feldman11} and \citet{KesselheimT17}. Very recently, \citet{charalampopoulos2025} devised a truthful mechanism which achieves  a constant approximation (for a very large constant) for \emph{posted pricing} and monotone submodular valuations. The state-of-the-art approximation for monotone or non-monotone submodular valuation functions for the strategic version we study here is due to \citet{amanatidis2022budget}, who achieve a competitive ratio of $1710$. Note that the Secretary problem (additive objective and cardinality constraint) has been studied in environments with predictions (see, e.g., \citet{antoniadis23,perchet23,fujii24,balkanski24fair}), yet our work is the first to study the online budget-feasible mechanism design problem---and its algorithmic counterpart---in such augmented environments.

\section{Preliminaries}
\label{sec:prelims}
We explore online mechanism design with predictions, focusing on a procurement auction setting involving a single auctioneer and multiple agents. The auctioneer has a valuation function $v : 2^N \rightarrow \mathbb{Q}_{\geq 0}$ and a budget $B>0$. We use $N = [n] = \{1, 2, \dots , n\}$ to denote the set of $n$ agents. Each agent $i$ has a \emph{private} cost $c_i> 0$, namely the cost of getting hired by the auctioneer. For $S \subseteq N$, $v(S)$ represents the value the auctioneer derives from selecting that set; for singletons, we write $v(i)$ instead of $v(\{i\})$. The algorithmic goal in all the problems we study is to select a set $S$ that maximizes $v(S)$ subject to the constraint $\sum_{i \in S} c_i \le B$. We assume oracle access to $v$ via value queries, implying the existence of a polynomial-time value oracle that returns $v(S)$ when queried with a set $S$.
A function $v$ is non-decreasing (referred to as \textit{monotone} herein), if $v(S) \leq v(T)$ for any $S \subseteq T \subseteq N$.
We consider the cases of monotone and non-monotone, normalised (i.e., $v(\emptyset) = 0$), non-negative
\emph{submodular} valuation functions. Since marginal values are extensively used, we adopt the shortcut $v(i | S)$ for the marginal value of agent $i$ with respect to the set $S$, i.e., $v(i | S) := v(S \cup \{i\}) - v(S)$.
\begin{definition}\label{def:1}
    A function $v$, defined on $2^N$ for some set $N$, is submodular if $v(i|S)\geq v(i|T)$ for all $S\subseteq T$ and $i\notin T$. 
\end{definition}
Note that when $v(i|S)=v(i|\emptyset)$, for all $i\in N$ and all $S\subseteq N$, the valuation function $v$ is additive. The following known result will be useful in our analysis in \Cref{sec:mono-submodular}.

\begin{theorem}[\protect\citet{NemhauserWF78}]
\label{thm:def1:iii}
    A function $v$ is submodular if and only if, for all $S,T\subseteq N$ 
    \[v(T)\leq v(S)+ \textstyle\sum_{i\in T\setminus S}  v(i|S)+  \textstyle\sum_{i\in S\setminus T}v(i|(S\cup T)\setminus\{i\}). \] 
\end{theorem}

An important element in our problem is the online arrival of agents. We adopt the \textit{random order} (RO) model; see also the discussion in the beginning of \Cref{sec:mono-submodular}. In this model an adversary first chooses the costs of all the agents in $N$.  There are $n$ different time slots $t_1,\dots, t_n$, with a single agent appearing in each slot according to a uniformly random permutation $\pi$, i.e.,  the  input sequence to an \emph{online} algorithm is ${\pi(1)},{\pi(2)}, \dots ,{\pi(n)}$. Such an  algorithm  has to perform its irrevocable actions for agent ${\pi(i)}$ before seeing agent ${\pi(i+1)}$. The length $n$ of the input sequence is known in advance. Note that this is necessary: without knowledge of $n$, no constant
competitive guarantee is achievable even for the classic, non-strategic, secretary
problem of which our setting is a generalization.\footnote{Intuitively, without knowing the
length of the input, an online algorithm cannot decide when to end its 
sampling phase. Therefore, an adversary controlling $n$ can make its
competitive ratio unbounded.} 
For the RO model, given the output $A$ of an algorithm and the value of an optimal solution $OPT$, we define the competitive ratio (for maximization problems) to be $\frac{OPT}{\mathbb{E}[v(A)]}$  on the adversarially chosen (worst case) set of inputs. The expectation is taken over all random permutations of the input. 
\medskip 

\noindent\textbf{Mechanism Design.}
An offline mechanism $M = (A, \vec{p})$ consists of an allocation function $A: \mathbb{R}_{\geq 0}^n \rightarrow 2^{[n]}$ and a payment rule $\vec{p}\!: \mathbb{R}_{\geq 0}^n \rightarrow \mathbb{R}_{\geq 0}^n$. Given a profile of cost declarations $\vec{b} \in [0, B]^n$, where each $b_i$ is the cost \emph{reported} by agent $i$, the allocation function selects a set of agents $A(\mathbf{b}) \subseteq N$. 
The payment rule returns a profile of payments $\vec{p}(\vec{b}) = (p_1, \dots, p_n)$, where $p_i(\vec{b}) \geq 0$ is the payment to agent $i \in N$. Therefore, the utility of agent $i \in N$ is $u_i(\vec{b}) = p_i(\vec{b}) - c_i$ if $i \in A(\vec{b})$ and $u_i(\vec{b})=0$ otherwise, with the understanding that agents not in $A(\vec{b})$ are not paid.

An \emph{online} mechanism, on the other hand,  decides the allocation and payment of an agent at every time step without having seen the entire input. 
That is, given a permutation $\pi$ on $N$, the mechanism decides about agent $\pi(t)$ at time step $t \in [n]$, as $\pi(t)$ appears and declares her cost, i.e., at the time the mechanism has only seen $\vec{b}_t^{\pi} = (b_{\pi(1)}, \dots, b_{\pi(t)}) \in [0, B]^t$.  The mechanism, thus, maintains the pair $(A_t(\vec{b}_t^{\pi}), p_t(\vec{b}_t^{\pi}))$ of currently selected agents and their payments, which is consistent with the decisions being irrevocable. 

For a \emph{deterministic} mechanism $M=(A, \vec{p})$ we require that, for every fixed arrival order, any true cost profile $\vec{c}$ and any declared cost profile $\vec{b}$, the mechanism is:
\begin{itemize}[itemsep=5pt, topsep=4pt, leftmargin=30pt]
    \item \emph{budget-feasible}; the payments to the agents do not exceed the auctioneer's budget, i.e., $ \sum_{i=1}^n p_i(\vec{b})\leq B.$
    \item \emph{individually rational}; no agent can be hired for less than they asked for, i.e., if $i\in A(\vec{b})$, then $p_i(\vec{b}) \geq b_i$.
    \item \emph{truthful}; no agent $i \in N$ has an incentive to misreport their true cost $c_i$, i.e., $u_i(c_i, \vec{b}_{-i}) \geq u_i(\vec{b})$.
\end{itemize}

A \emph{randomized} mechanism $M = (A, \vec{p})$ can be thought of as a probability distribution over deterministic mechanisms for every fixed arrival order. In this work, we require all our randomized mechanisms to be probability distributions over budget-feasible, individually rational, and truthful deterministic mechanisms. Randomized mechanisms with the latter property are called \emph{universally truthful}.\footnote{This is stronger than \emph{truthfulness in expectation}, which requires that truth-telling maximizes each agent's expected utility.} Since all our mechanisms are provably universally truthful, we will consistently refer to the true cost profile $\vec{c}$ instead of the declared cost profile $\vec{b}$ when needed. Furthermore, we will omit referring explicitly to $\vec{c}$ if it is clear from the context.

Given the above, an instance of the problem is described by the tuple $I=(N, \vec{c}, v,B)$. A common practice in this line of work is to measure the performance of a mechanism $M=(A, \vec{p})$, both in online variants and the offline problem, against the optimal value of the following packing problem:\vspace{-3pt}
\begin{equation}\label{eq:opt}
    \max\ v(S) \quad \text{s.t.} \quad \textstyle\sum_{i \in S }c_i \leq B, \quad S \subseteq N.
    \vspace{-3pt}
\end{equation}
We use $S^*(\vec{c}) \subseteq N$ to refer to an optimal solution of \eqref{eq:opt}.

\medskip

\noindent
\textbf{Predictions.} We employ a beyond-worst-case analysis perspective to measure the performance of mechanisms. We assume that each instance $I = (N, \vec{c}, v, B)$ is \emph{augmented} by a (possibly erroneous) deterministic prediction $\omega$ of $v(S^*(\vec{c}))$ and 
we denote the augmented instance as $I^+ =  (N,\allowbreak \vec{c}, v, B, \omega)$.
As is standard in the algorithms with predictions literature, we assess the performance of a mechanism {$M = (A, \vec{p})$}
on aug\-mented instances based on the following two metrics:
\begin{itemize}[itemsep=5pt, topsep=4pt, leftmargin=30pt]
    \item \emph{Consistency:} $M $ is $\alpha$-consistent  if $\alpha \cdot \EX\left[v(A(\vec{c}))\right]  \geq v(S^*(\vec{c}))$, for every instance with a perfect prediction, i.e., when $\omega = v(S^*(\vec{c}))$. 
    \item \emph{Robustness:} $M $ is $\rho$-robust  if $\rho \cdot \EX\left[v(A(\vec{c}))\right]  \geq  v(S^*(\vec{c})),$ for every instance with an arbitrary prediction $\omega \in (0, v(S^*(\vec{c}))]$.
\end{itemize}
Clearly, $\alpha, \rho \ge 1$. Also, note that the expectation in each of the two benchmarks above accounts for  both the randomness of $M$ and the randomness of the RO model. In addition to consistency and robustness, we also assess the performance of mechanisms when the prediction is ``approximately'' correct. To this end, we introduce an \emph{error parameter} $ \varepsilon$ that quantifies how far the prediction is from $v(S^*(\vec{c}))$. Formally, for an instance $I^+$, we define the constant $\varepsilon \in [0, 1)$ to be such that $\omega = (1 - \varepsilon) \cdot v(S^*(\vec{c}))$, or equivalently, $\varepsilon= 1 - \omega / v(S^*(\vec{c}))$. 

Note that, throughout this work, we assume that $\omega$ is an underestimator of $v(S^*)$ and, thus,
the analyses presented  refer to a parameterization of the prediction as $\omega=(1-\varepsilon)\cdot v(S^*)$ with $\varepsilon\in [0,1)$. This is mainly for the sake of presentation; our results extend for  $\omega=(1+\varepsilon)\cdot v(S^*)$ with $\varepsilon\in [0,\kappa]$, for constant $\kappa$, by using $\hat{\omega}=\omega/(\kappa + 1)\leq v(S^*)$.
However, extending to arbitrary overestimators is not as straightforward, as it would require an additional initial sampling step; we consider this an interesting direction for future work.

 \section{A Natural Prediction-Augmented Mechanism}\label{sec:mono-submodular}
In this section, we assume that the valuation function $v$ is monotone submodular. We address the more general, non-monotone case in \Cref{sec:nonmono-submodular}.
 
We start with a short discussion about the online model we consider. The most extensively studied model for online algorithms is the adversarial one. For our problem, that would mean that an adversary chooses the cost and value of each agent, as well as the time of their arrival. Versus such an adversary though, no strong approximation guarantee for consistency and robustness simultaneously is possible. In fact, Proposition \ref{proposition:adversarial} is about the approximation guarantee of mechanisms in the adversarial arrival model; to remain in line with the algorithms with the predictions theme, we refer to it as robustness.
\begin{restatable}[name=]{proposition}{restatedpropadversarial}
 \label{proposition:adversarial}
   No randomized online mechanism, augmented with a prediction for the value of an optimal solution, can  have in-expectation $o(n)$-robustness in the adversarial model of arrival.  
 \end{restatable}
\begin{proof}
For the sake of contradiction, assume that there exists a mechanism $M$  that can achieve constant expected robustness guarantee in the adversarial model of arrival. Now, consider the following family of $n$ instances $\{I_1, \dots, I_n\}$. Each instance $I_k$ consists of $n$ agents, where every agent $i$ has a cost $c_i = B$, their values are defined based on their arrival order  as follows: 
\begin{itemize}
    \item $v_i = n^i$ for $i \in \{1,\dots,k\}$ 
    \item $v_i = 0$ for $i \in \{k+1,\dots,n\}$. 
\end{itemize}
Since at most one agent can be hired, for any instance $I_k$, the optimal solution is to select the agent in position $k$. 
Since mechanism  $M$ processes agents online, its decision to select the agent in position $i$ depends only on the values observed in positions $1,\dots, i$. Let $q_i$ be the probability that mechanism $M$ hires the $i$-th agent upon observing the values $n, n^2, \dots, n^i$, conditional on not having hired anyone at positions $1,\ldots,i-1$. 
Since the values grow exponentially in $n$, in order for a mechanism $M$ to achieve $o(n)$ in-expectation robustness, the probability $q_i$ needs to be $\omega({1}/{n})$ for every $i\in [n]$. 
However, the latter causes the inequality $\sum_{i=1}^{n}q_i \le 1$---which holds because at most one agent is hired---to fail asymptotically, leading to a contradiction. 
\end{proof}
 
 Given \Cref{proposition:adversarial}, we turn our attention on the RO model of arrival, defined in \Cref{sec:prelims}. The main result of this section is \Cref{mechanism:mono-sm-conv}, which is a universally truthful, budget-feasible randomized mechanism with constant consistency and robustness guarantees that runs in polynomial time. Moreover, the tradeoff between consistency and robustness that the mechanism achieves is tunable via an error tolerance parameter $\tau \in [0,1)$: larger $\tau$ favors consistency, while smaller $\tau$ favors robustness. \Cref{mechanism:mono-sm-conv} can be viewed as a convex combination between two other mechanisms, namely \Cref{mechanism:mono-sm-pred} and \Cref{mechanism:mono-sm-withoutpred}. The first one (\Cref{mechanism:mono-sm-pred}) is a posted price mechanism that builds a budget-feasible solution based solely on the information provided by the prediction. The second one (\Cref{mechanism:mono-sm-withoutpred}) is also a posted price mechanism which, however, ignores the prediction. This mechanism bears similarities to the mechanism of \citet{badanidiyuru2012learning}, however, our analysis is significantly tighter.
 The resulting mechanism (\Cref{mechanism:mono-sm-conv}) randomizes between the two outcomes using a tolerance parameter $\tau \in [0,1)$.

 Parameters $a\in (0,1),$ $\beta\in (0,1)$, $\gamma \in [1, \infty)$, $p\in [0,1]$, $q\in[0,1]$ are parameters of the mechanisms, whereas $z\in [1,\infty)$ and $\delta\in (0,1)$  are auxiliary parameters that appear in the statements and the analyses of the mechanisms (hence their value can be set independently for each mechanism). Specific values for all parameters are provided in the corollaries that follow each main theorem. This applies to all the mechanisms that will be presented in this paper.

\setlength{\algomargin}{9pt}
\makeatletter
\patchcmd{\@algocf@start}
  {-1.5em}
  {2pt}
{}{}
\makeatother


\begin{mechanism}[!h]
\caption{Online mechanism for a monotone submodular $v$, parameterized by tolerance $\tau \in [0,1)$.
}\label{mechanism:mono-sm-conv}
\DontPrintSemicolon
\setstretch{1.1}
\SetInd{4pt}{7pt} 
\begin{algolabel}{}[\textbf{With probability} $\tau$\,\textbf{:}]
Run \Cref{mechanism:mono-sm-pred} \label{line:Mechanism-pred}

\end{algolabel}

\begin{algolabel}{}[\textbf{With probability $1-\tau$\,:}]

Run \Cref{mechanism:mono-sm-withoutpred} \label{line:Mechanism-withoutpred}
\end{algolabel}

\end{mechanism}

\setlength{\algomargin}{9pt}
\makeatletter
\patchcmd{\@algocf@start}
  {-1.5em}
  {2pt}
{}{}
\makeatother


\begin{mechanism}[!ht]
\caption{
Online mechanism for a monotone submodular function $v$, augmented with a prediction $\omega$ of $v(S^*)$.
}\label{mechanism:mono-sm-pred}
\DontPrintSemicolon
\setstretch{1.1}
\SetInd{4pt}{7pt} 
\begin{algolabel}{}[\textbf{With probability} $p$\,\textbf{:} ]

Return the first agent $j$ for whom $v(j)\geq \frac{a\cdot \omega}{\gamma}$ and pay her $B$.
\end{algolabel}

\begin{algolabel}{}[\textbf{With probability $1-p$\,:}] 
Set $t=a\cdot\omega$

Set $S=\emptyset$, $B'=B$

\For{each round  as an agent $i$ arrives}{

\If{$c_i \leq \bar{p}_i:=\frac{B}{t} v(i \,|\, S) $ and $B'-\bar{p}_i \geq 0$}{

    Add $i$ to $S$, set $p_i=\bar{p}_i$ and $B'=B'-p_i$.
}
}
\end{algolabel}

\Return winning set $S$ and payments $\vec{p}$.
\end{mechanism}\smallskip

\setlength{\algomargin}{9pt}
\makeatletter
\patchcmd{\@algocf@start}
  {-1.5em}
  {2pt}
{}{}
\makeatother


\begin{mechanism}[!h]
\caption{Online mechanism for a monotone submodular function $v$.}\label{mechanism:mono-sm-withoutpred}
\DontPrintSemicolon
\setstretch{1.1}
\SetInd{4pt}{7pt} 
\begin{algolabel}{} [\textbf{With probability} $q$\,\textbf{:} \SetNoFillComment {\small \tcc*[f]{Dynkin's Algorithm}}] 

Sample the first $\floor{n/e}$ agents; let $i^*$ be the most valuable among them.

From the remaining agents, return the first agent $j$ for whom $v(j)\geq v(i^*)$ and pay her $B$.
\end{algolabel}

\begin{algolabel}{}[\textbf{With probability $1-q$\,:}]

Draw $\xi_1$ from the binomial distribution $\mathcal{B}(n,\frac{1}{2})$.

Let $N_1$ be the set of the first $\xi_1$ agents that arrive.  \

Let $T_1$ be a $(1-\frac{1}{e})$-approximate solution to \eqref{eq:opt} on $N_1$, using the algorithm of \citet{sviridenko2004}. 

Set $N_2=N\setminus N_1, B'=B,S=\emptyset$ and $\vec{p}=\vec{0}.$ 

Set $t=\beta\cdot v(T_1)$

\For{each round  as agent $i \in N_2$ arrives}{

\If{$c_i \leq \bar{p}_i:=\frac{B}{t} v(i \,|\, S) $ and $B'-\bar{p}_i \geq 0$}{

    Add $i$ to $S$, set $p_i=\bar{p}_i$ and $B'=B'-p_i$.
}
}
\end{algolabel}

\Return winning set $S$ and payments $\vec{p}$.
\end{mechanism}\smallskip

Before analyzing the approximation performance, we first establish all the other desirable properties of \Cref{mechanism:mono-sm-conv} in the following lemma.
\begin{lemma}
\label{lem:truth-bf 1}
     \Cref{mechanism:mono-sm-conv} is universally truthful, budget-feasible and individually rational.
\end{lemma}
\begin{proof}
 Regarding universal truthfulness, since we have a probability distribution over two different mechanisms, it suffices to argue about the truthfulness of each of them separately.

 Firstly, in \Cref{mechanism:mono-sm-pred} with probability $p$, we run a procedure, that only hires the first agent $j$ such that $v(j)\geq \frac{a\cdot \omega}{\gamma}$ and pay her $B$, which is truthful since the private information of the agents is not used and for each agent $i$ it holds that $c_i\leq B$.  

Coming now to the second part of \Cref{mechanism:mono-sm-pred} that is run with probability $1-p$,  fix an arrival sequence, as realized by the random arrival model.
Note that the agents have no control over the position in which they arrive. 
Moreover, each agent $i \in N$ is offered a price $p_i$, which is independent of her declared cost, since agent $i$ has no control over her slot of arrival and thus cannot affect her marginal contribution.

For the sake of contradiction, assume that agent $j\in N$ can achieve a better outcome by misreporting a cost $b_j \neq c_j$, while all other agents keep the same reported costs. Suppose first that agent $j$ belongs to the winning set $S$, under the truthful profile.
Then by misreporting, she receives the same payment as long as her declared cost is below the threshold $p_j$ which is still the same as before, since the algorithm runs in exactly the same way up until the time that $j$ arrives.  
If the declared cost is above the threshold, she is rejected by the mechanism and receives a payment of 0. Hence, when $j \in S$, she cannot guarantee a better utility by misreporting her true cost. Suppose now that agent $j$ does not belong to the solution $S$, under truthful reporting. She certainly cannot benefit by reporting a higher cost than $c_j$, as that would still result in rejection by the mechanism.
Additionally, if she reports a lower cost than $c_j$ and she is added to the solution $S$, this means that the reason for rejection before was that the offered payment was less than her true cost $c_j$ (and not because of budget exhaustion). By reporting $b_j < c_j$, given that the payment offered by the mechanism remains the same, this results in negative utility for agent $j$. All of the above leads to the conclusion that no agent can benefit from misreporting her true cost.

In a similar manner, one can prove that \Cref{mechanism:mono-sm-withoutpred} is also truthful. Note that any agent $i \in N_1$ is always rejected by the mechanism, regardless of the cost she declares. 

For budget feasibility, note that for both mechanisms, in the case that  the mechanism hires the first agent that passes threshold $\frac{a\cdot \omega}{\gamma}$ (\Cref{mechanism:mono-sm-pred}) or when Dynkin's algorithm is run (\Cref{mechanism:mono-sm-withoutpred}), the mechanism pays exactly $B$. Moreover, in both mechanisms the inequality $B'-\bar{p_i}\geq 0$  guarantees that for the solution $S$ returned by the mechanism, $\sum_{i \in S} p_i \leq B$, which ensures the budget feasibility of the solution $S$, in the case that Dynkin's algorithm is not selected.

Finally, through the inequality $c_i\leq \bar{p_i}$ and the fact that $c_i\leq B$ for any $i\in N$, individual rationality is ensured for both mechanisms, since for any agent $i$ added to $S$, it holds that $p_i\geq c_i$. 
\end{proof}

The main theorem of this section is \Cref{thm:approx-ratio-conv-mech} below, which provides the consistency and robustness guarantees, in terms of the involved parameters.
\begin{theorem} \label{thm:approx-ratio-conv-mech}
    \Cref{mechanism:mono-sm-conv}  with access to a prediction $\omega=(1-\varepsilon)\cdot v(S^*)$ with $\varepsilon\in [0,1)$, achieves an in-expectation approximation guarantee of $(\tau \cdot f_1+(1-\tau)\cdot f_2)^{-1}$, where $f_1$  is the expected approximation ratio of \Cref{mechanism:mono-sm-pred} and is equal to 
    \[\min\!\left\{    \frac{pa (1-\varepsilon)}{\gamma}, \allowbreak \frac{a(1-\varepsilon)(1-p)(\gamma}-1){\gamma}
    ,(1-p)(1-a(1-\varepsilon))\right\} \]
    and $f_2$ is the expected approximation ratio of \Cref{mechanism:mono-sm-withoutpred}, which is
    \[\min\left\{
  \frac{q\,\beta\,(e-1)(1-2\delta)}{2e^{2}z},
  \frac{(1-q)\,\tilde{p}\,(1-2\delta-2\beta)}{2},
  \frac{(1-q)\,\tilde{p}\,(z-1)\,\beta\,(e-1)(1-2\delta)}{2ez}
\right\},\]
    with $\tilde{p}= 1 - 2\exp\!\left(
    -\frac{12\delta^{2}ez}{\beta\,(e-1)\,(1-2\delta)\,(3+4\delta)}\right).$ \\
\end{theorem}
From \Cref{thm:approx-ratio-conv-mech}, we obtain the following two corollaries, which give us the in-expectation consistency and robustness guarantees of \Cref{mechanism:mono-sm-conv}. 
\begin{corollary}\label{corolarry:consistency-conv-mech}
    \Cref{mechanism:mono-sm-conv}, given access to a perfect prediction $\omega=v(S^*)$ achieves an in-expectation consistency factor of $(\tau\cdot \frac{1}{6}+(1-\tau)\cdot \frac{1}{146})^{-1}$, by setting $p=0.46$, $a=0.685$, $\gamma=1.85$ for \Cref{mechanism:mono-sm-pred} and $q=0.66$, $z=2.1$, $\beta=0.29$, $\delta=0.174$ for \Cref{mechanism:mono-sm-withoutpred}. 
\end{corollary}
\begin{corollary}
    \Cref{mechanism:mono-sm-conv}, given access to an erroneous prediction $\omega=(1-\varepsilon)\cdot v(S^*)$, with $\varepsilon\in (0,1)$, achieves an in-expectation robustness factor of $(\tau\cdot (\frac{1}{6}\cdot(1-\varepsilon))+(1-\tau)\cdot \frac{1}{146})^{-1}$ by setting $p=0.46$, $a=0.685$, $\gamma=1.85$ for \Cref{mechanism:mono-sm-pred} and $q=0.66$, $z=2.1$, $\beta=0.29$, $\delta=0.174$ for \Cref{mechanism:mono-sm-withoutpred}.
\end{corollary}
Below, we will present the proof of \Cref{thm:approx-ratio-conv-mech}, by analysing separately the in-expectation performance guarantees of \Cref{mechanism:mono-sm-pred} and \Cref{mechanism:mono-sm-withoutpred}. From now and for the rest of the paper, we will denote as $S^* = \{a_1,\dots, a_\ell\}$ the winning set of an optimal offline solution (arbitrarily chosen if there is more than one) and also $S_1^*=N_1\cap S^*$, $S_2^*=N_2\cap S^*$ (for those mechanisms that use sets $N_1$ and $N_2$). 
\begin{lemma}\label{lem:mechanism-with-prediction}
    \Cref{mechanism:mono-sm-pred}, with access to a prediction $\omega=(1-\varepsilon)\cdot v(S^*)$ with $\varepsilon\in [0,1)$,  has an approximation ratio of $\min\left\{
  \frac{pa(1-\varepsilon)}{\gamma}, \,
  (1-p)\bigl(1 - (1-\varepsilon)a\bigr), \,
  \frac{a(1-\varepsilon)(1-p)(\gamma-1)}{\gamma}
\right\}$, in expectation. 
\end{lemma}
\begin{proof}

Let $i^*\in \argmax_j v(\{j\})$. We will discriminate between three cases over all possible outcomes.
\smallskip
 
\noindent\textbf{Case 1: }\underline{Agent $i^*$ has significant value, i.e.: $v(i^*)>\frac{a\cdot (1-\varepsilon)\cdot  v(S^*)}{\gamma}$.}\\[4pt]
For this case, we will utilize the fact that with probability $p$, the mechanism hires the first agent $j$ for whom it holds that $v(j)\geq \frac{a\cdot \omega}{\gamma}=\frac{a\cdot (1-\varepsilon)\cdot v(S^*)}{\gamma}$. Therefore, independently of what our mechanism does with the remaining probability, the expected value for $v(S)$ will be at least 
\[E[v(S)] \geq  \frac{p\cdot a\cdot (1-\varepsilon)}{\gamma} \cdot v(S^*).\]
This concludes the first case.
\smallskip

\noindent\textbf{Case 2: }\underline{Each agent $i\in S^*\setminus S$ rejects the posted price and $v(i^*)\leq \frac{a\cdot(1-\varepsilon)\cdot v(S^*)}{\gamma} $.}\\[4pt] 
An agent $i$ rejects the posted price $\bar{p_i}$, only if $c_i>\bar{p}_i$. Then, given that $S_i$ is the current solution at the time when $i$ arrives,
\begin{align*}
        &\sum\limits_{i\in S^*\setminus S}(v(S_i\cup\{i\})-v(S_i))<\sum\limits_{i\in S^*\setminus S}\frac{t}{B}\cdot c_i\leq \frac{a\cdot \omega}{B}\cdot B=a\cdot (1-\varepsilon)\cdot  v(S^*),
\end{align*}
where the first inequality holds, due to the fact that each agent in $S^*\setminus S$ rejected the posted price offered by the mechanism and the second from the fact that the optimal solution is feasible. This inequality combined with \Cref{thm:def1:iii}, leads us to the following one :
\begin{gather*}
    v(S^*)-v(S)\leq \sum\limits_{i\in S^*\setminus S}(v(S_i\cup\{i\})-v(S_i))\leq a\cdot (1-\varepsilon)\cdot  v(S^*)\implies \\ v(S)\geq v(S^*)-a\cdot (1-\varepsilon)\cdot  v(S^*)= (1-a\cdot(1-\varepsilon))\cdot v(S^*)
\end{gather*}
Since \Cref{mechanism:mono-sm-pred} runs this part with probability $1-p$, in expectation we have a factor of at least
\[(1-p)\cdot (1-a\cdot(1-\varepsilon)).\]

\noindent\textbf{Case 3:} \underline{There is an agent $j\in S^*\setminus S$ such that $c_j\leq \bar{p}_j$ but also $\bar{p}_j>B'$ and $v(i^*)\leq \frac{a\cdot(1-\varepsilon)\cdot v(S^*)}{\gamma} $.}\\[4pt]
For any agent $i$ from $S^*\setminus S$ who accepted the posted price, it must hold that $ \bar{p}_i > B'$, where $B'$ is the remaining budget at the round in which agent $i$ is considered. 
We first show that  $B'\leq \frac{B}{\gamma}$. To see this:
    \begin{align*}
         B'<\bar{p}_j=\frac{B}{t}(v(S_j\cup \{j\})-v(S_j))\leq \frac{B}{t} v(j)\leq\frac{B}{a\cdot (1-\varepsilon)\cdot v(S^*)} v(j)\leq \frac{B}{\gamma}, 
    \end{align*}
 where the last inequality holds due to the fact that $v(j)\leq v(i^*)\leq \frac{a\cdot (1-\varepsilon)}{\gamma} v(S^*)$. We have also used submodularity, since $v(j|S_j)\leq v(j)$.

\noindent Since $B'\!\leq\! \frac{B}{\gamma}$, we have that $\sum\limits_{i\in S} \bar{p}_i\geq B-\frac{B}{\gamma}=\frac{(\gamma-1)\cdot B}{\gamma}$. Then,
    \begin{align*}
        \frac{\gamma-1}{\gamma}\cdot B\leq\sum\limits_{i\in S} \bar{p}_i=\sum\limits_{i\in S}\frac{B}{t}(v(S_i\cup \{i\})-v(S_i))  \ =\  B\cdot \frac{v(S)}{a\cdot\omega}=B\cdot \frac{v(S)}{a\cdot(1-\varepsilon)\cdot v(S^*)},
    \end{align*}
    which leads to an expected consistency factor of  
    $$(1-p)\cdot a\cdot (1-\varepsilon)\cdot   \frac{\gamma-1}{\gamma}.$$

Combining the analysis of all the three cases above, \Cref{mechanism:mono-sm-pred} achieves an in-expectation consistency of at least:
\[ \min\!\left\{
    \frac{a (1-\varepsilon) p}{\gamma}, \frac{a (1-\varepsilon)(1-p)(\gamma-1)}{\gamma},(1-p)(1-a(1-\varepsilon))\right\}.
\]
This concludes the proof.
\end{proof}
Before continuing, we state Bernstein's concentration inequality (\Cref{thm:bernstein}) for bounded random variables as it will be useful for the analysis that follows.
We note that we use this particular concentration inequality since the sum of the random variables we define within the proof has bounded variance.

\begin{theorem}[Bernstein's Inequality]\label{thm:bernstein} 
 Let $X_1, \dots, X_n$ be independent random variables such that almost surely $a_i\leq X_i\leq b_i$ and $b_i-a_i\leq C$, for any $i$. Consider the sum of these random variables, $S_n=\sum\limits_{i=1}^n X_i$. Then, for any  $\delta> 0$ it holds that
\begin{align*}
        Pr[|S_n-\EX[S_n]|\geq \delta]\leq 2 \exp{\Big(-\frac{\delta^2/ 2 }{V_n+C\cdot \delta /3}\Big)},
\end{align*}
where $\EX[S_n]$ is the expected value of $S_n$ and $V_n$ is the variance of $S_n$.   
\end{theorem}
\Cref{lem:Hoeffding} will help us bound the quantity of the optimal solution that is included in sets $S_1^*$ and $S_2^*$ for \Cref{mechanism:mono-sm-withoutpred}.
\begin{restatable}{lemma}{lemHoeffding}\label{lem:Hoeffding}
When $v(i^*) \leq \frac{\beta}{z}\cdot\frac{e-1}{e}\cdot(\frac{1}{2}-\delta)\cdot v(S^*)$, for all $i\in N$ then with
probability at least $1 - 2 \exp\!\big(
  - \frac{12 \delta^{2} e z}{\beta (e-1)(1-2\delta)(3+4\delta)}
\big)$, it holds that $(\frac{1}{2}-\delta)\cdot v(S^*)\leq v(S_1^*)$, $(\frac{1}{2}-\delta)\cdot v(S^*)\leq v(S_2^*) $, where $\delta$ is a positive constant, that we will set to a certain value later.
 \end{restatable}
\begin{proof}
Fix $S^* = \{a_1,\dots, a_\ell\}$ to be the winning set of an optimal offline solution (arbitrarily chosen if there is more than one). Let $S_1^*=N_1\cap S^*$, $S_2^*=N_2\cap S^*$. We can rearrange the agents in the optimal solution $S^*= \{a_1,\dots, a_\ell\}$, so that they are sorted according to decreasing marginal contributions, i.e., $i\in \argmax_j v(\{a_1,\dots, a_{i-1}\} \cup \{j\}) - v(\{a_1,\dots, a_{i-1}\})$ for $i>1$, and for $i=1$, $v(a_1) = \max_j v(\{a_j\})$.
Let $w_i$ denote the marginal contribution of $a_i$ w.r.t the set $\{a_1,\dots, a_{i-1}\}$.
Since the agents are assumed to arrive in a uniformly random order, we have that each agent is in $N_1$ with probability $\frac{1}{2}$ independently of other agents. Consider the random variables $X_1, \dots, X_\ell$ corresponding to the agents of the optimal solution $S^*$, defined as follows: for $i=1,\dots, \ell$, $X_i=w_i$ when $a_i \in N_1$ and $X_i=0$ otherwise. The previous discussion implies that $X_i$ takes the value $w_i$ with probability $\frac{1}{2}$. Next, we define $X = \sum_{i=1}^\ell X_i$. Symmetrically, consider the random variables $\bar{X_1}, \dots, \bar{X_\ell}$, where  $\bar{X_i}=w_i$ when $a_i \in N_2$ and $\bar{X_i}=0$ otherwise and define $\bar{X}=\sum_{i=1}^\ell \bar{X_i}$.  Observe that $\EX[X]=\frac{1}{2}\cdot v(S^*)$ and due to submodularity, $v(S_2^*) \geq \bar{X}$. Moreover, we have that
\begin{equation}\label{eq:var-ub}
\mathrm{Var}[X]=\sum\limits_{i=1}^{\ell}\mathrm{Var}[X_i]=\sum\limits_{i=1}^{\ell}\frac{w_i^2}{4}\leq \sum\limits_{i=1}^{\ell}\frac{w_i v(i^*)}{4}=\frac{v(S^*) v(i^*)}{4}\leq \frac{\beta}{4 z}\frac{e-1}{e}(\frac{1}{2}-\delta) v(S^*)^2    
\end{equation}
The first equality follows from the fact that $X_1,\dots,X_{\ell}$ are independent. Then, the first inequality holds since $v(i^*)=\max_{i} w_i$, whereas the second inequality follows by the assumed upper bound on $v(i^*)$. Finally, it holds that $X+\bar{X} = \sum_{i\in [\ell]} w_i=  v(S^*)$.
This implies that $v(S_2^*)+X\geq v(S^*)\Rightarrow v(S_2^*)\geq v(S^*)-X$. 

\noindent Repeating that  $X = \sum\limits_{i=1}^\ell X_i$, with $\mathbb{E}[X]=\frac{1}{2}v(S^*)$ it holds that $v(S_1^*) \geq X$ and $v(S_2^*)\geq v(S^*) - X$,  due to submodularity. Let $k=\frac{X}{v(S^*)}$ with $\EX[k]=\frac{1}{2} $ and $\mathrm{Var}[k]=\frac{\mathrm{Var}[X]}{v(S^*)^2}$. We also have that
\begin{align*}
   Pr(|k-\EX[k]|\geq \delta) &\leq 2 \exp{\bigg(-\frac{\delta^2/ 2 }{\mathrm{Var}[k]+\frac{\beta}{z}\cdot\frac{e-1}{e}\cdot(\frac{1}{2}-\delta)\cdot \frac{\delta}{3}}\bigg)}\\& \leq 2 \exp{\bigg(-\frac{\delta^2/ 2 }{(\frac{\beta}{4\cdot z}\cdot\frac{e-1}{e}\cdot(\frac{1}{2}-\delta)+\frac{\beta}{z}\cdot\frac{e-1}{e}\cdot(\frac{1}{2}-\delta)\cdot \frac{\delta}{3})}\bigg)}=\\& =2 \exp{\bigg(-\frac{\delta^2/ 2 }{(\frac{\beta}{z}\cdot\frac{e-1}{e}\cdot(\frac{1}{2}-\delta))\cdot(\frac{1}{4}+ \frac{\delta}{3})}\bigg)}.
\end{align*}
Here, the first inequality follows from  \Cref{thm:bernstein}, while the second inequality is due to \eqref{eq:var-ub}. By rearranging, we obtain
\begin{equation*}
    Pr(\EX[k]-\delta \leq k\leq \EX[k]+\delta)\geq 1-2 \exp{\bigg(-\frac{\delta^2/ 2 }{(\frac{\beta}{z}\cdot\frac{e-1}{e}\cdot(\frac{1}{2}-\delta))\cdot(\frac{1}{4}+ \frac{\delta}{3})}\bigg)}.
\end{equation*}
\noindent Hence, with probability at least  $1-2 \exp{\big(-\frac{\delta^2/ 2 }{(\frac{\beta}{z}\cdot\frac{e-1}{e}\cdot(\frac{1}{2}-\delta))\cdot(\frac{1}{4}+ \frac{\delta}{3})}\big)}$, it holds that $(\frac{1}{2}-\delta)\cdot v(S^*)\leq X\leq (\frac{1}{2}+\delta)\cdot v(S^*)$. Therefore, $v(S_1^*)\geq(\frac{1}{2}-\delta)\cdot v(S^*) $ and $v(S_2^*)\geq v(S^*)-(\frac{1}{2}+\delta)\cdot v(S^*)=(\frac{1}{2}-\delta)\cdot v(S^*)$, also hold with probability at least $1-2 \exp{\left(-\frac{\delta^2/ 2 }{(\frac{\beta}{z}\cdot\frac{e-1}{e}\cdot(\frac{1}{2}-\delta))\cdot(\frac{1}{4}+ \frac{\delta}{3})}\right)}$, which completes the proof.

\end{proof}

We can now prove an approximation guarantee of \Cref{mechanism:mono-sm-withoutpred}.
\begin{lemma}\label{lem:mechanism-without-prediction}
    \Cref{mechanism:mono-sm-withoutpred}, has an approximation ratio of:
    \begin{align*} 
    \min\left\{
  \frac{q\,\beta\,(e-1)(1-2\delta)}{2e^{2}z},
  \frac{(1-q)\,\tilde{p}\,(1-2\delta-2\beta)}{2},
  \frac{(1-q)\,\tilde{p}\,(z-1)\,\beta\,(e-1)(1-2\delta)}{2ez}
\right\}
     \end{align*}
     in expectation, where $\tilde{p}= 1 - 2\exp\!\left(
    -\frac{12\delta^{2}ez}{\beta\,(e-1)\,(1-2\delta)\,(3+4\delta)}\right).$   
\end{lemma}
\begin{proof}
    We will discriminate between the same three cases that we did for the analysis of \Cref{mechanism:mono-sm-pred}.
\smallskip

\noindent\textbf{Case 1:} \underline{Agent $i^*$ has significant value, i.e.: $v(i^*) > \frac{\beta}{z}\cdot \frac{e-1}{e}\cdot (\frac{1}{2}-\delta) \cdot v(S^*)$.}\\[4pt] 
For this case, we will utilize the fact that with probability $q$ we run Dynkin's algorithm, which gives a  factor of $1/e$ over single-agent solutions. Therefore, independently of what our mechanism does with the remaining probability, the expected value for $v(S)$ will be at least:
\[E[v(S)] \geq \frac{q}{e}\cdot \frac{\beta}{z}\cdot \frac{e-1}{e}\cdot \big(\frac{1}{2}-\delta\big)\cdot v(S^*)= \frac{q\,\beta\,(e-1)\,\left(1 - 2\delta\right)}{2e^{2}z}\cdot v(S^*).\]

\noindent\textbf{Case 2:} \underline{Each agent $i\in S_2^*\setminus S$ rejects the posted price and $v(i^*) \leq \frac{\beta}{z}\cdot \frac{e-1}{e}\cdot (\frac{1}{2}-\delta) \cdot v(S^*)$.}\\[4pt]  
Then,
\begin{align*}
        \sum\limits_{i\in S_2^*\setminus S}(v(S_i\cup\{i\})-v(S_i))<\sum\limits_{i\in S_2^*\setminus S}\frac{t}{B}\cdot  c_i = \frac{\beta\cdot v(T_1)}{B}\cdot B \le \beta \cdot  v(S^*),
\end{align*}
where the first inequality holds, due to the fact that each agent in $S_2^*\setminus S$ rejected the posted price offered by the mechanism and the second from the fact that the partial solution from $N_1$ is upper bounded by the optimal solution. This inequality leads us to the following one:
\begin{align*}
    v(S_2^*)-v(S)\le \beta \cdot v(S^*)\Leftrightarrow v(S)\geq v(S_2^*)-\beta \cdot v(S^*).
\end{align*}
By plugging the bound of \Cref{lem:Hoeffding} into the above equation and rearranging the terms we get
\[v(S)\geq \frac{1-2\delta -2 \beta}{2} \cdot v(S^*),\]
with probability at least $1-2 \exp{\left(-\frac{\delta^2/ 2 }{(\frac{\beta}{z}\cdot\frac{e-1}{e}\cdot(\frac{1}{2}-\delta))\cdot(\frac{1}{4}+ \frac{\delta}{3})}\right)}$. \\

\noindent\textbf{Case 3:} \underline{There is a $j\in S_2^*\setminus S$ such that $c_j\leq\bar{p}_j$ but $\bar{p}_j>B'$ and $v(i^*) \leq \frac{\beta}{z}\cdot \frac{e-1}{e}\cdot (\frac{1}{2}-\delta) \cdot v(S^*)$.}\\[4pt] 
Let $j\in S_2^*\setminus S$ be the first agent that has this property.

\noindent We first show that  $B'\leq \frac{B}{z}$. To see this:
    \begin{align*}
        B'< \bar{p}_j =\frac{B}{t}\cdot(v(S_j\cup \{j\})-v(S_j)) \le \frac{B}{t}\cdot v(j)\leq \frac{B\cdot v(j)}{\beta \cdot v(T_1)} \leq \frac{B}{z}, 
    \end{align*}
 where the last inequality holds due to the fact that $v(j)\leq v(i^*)\leq \frac{\beta}{z}\cdot \frac{e-1}{e}\cdot (\frac{1}{2}-\delta)\cdot  v(S^*)\leq \frac{\beta \cdot v(T_1)}{z}$ and the event of \Cref{lem:Hoeffding}. We have also used submodularity, since $v(j|S_j)\leq v(j)$.

\noindent Since $B'\leq \frac{B}{z}$, we have that, $\sum\limits_{i\in S} p_i\geq B-\frac{B}{z}=\frac{(z-1)\cdot B}{z}$ and also:
    \begin{align*}
        \frac{z-1}{z}\cdot B\leq\sum\limits_{i\in S} p_i=\sum\limits_{i\in S}\frac{B}{t}(v(S_i\cup \{i\})-v(S_i))\leq B\cdot \frac{v(S)}{\beta \cdot v(T_1)},
    \end{align*}
    which leads us to $v(S)\geq \frac{(z-1)}{z}\cdot \beta \cdot v(T_1)\geq \frac{(z-1)\beta (e-1)}{ez}\cdot  v(S_1^*)$. This, via \Cref{lem:Hoeffding}, gives us that  $v(S)\geq \frac{(z-1)\beta (e-1)}{ez}\cdot (\frac{1}{2}-\delta) \cdot v(S^*) =\frac{(z-1)\beta(e-1)(1-2\delta)}{2ez}\cdot v(S^*)$ with probability at least $1-2 \exp{\left(-\frac{\delta^2/ 2 }{(\frac{\beta}{z}\cdot\frac{e-1}{e}\cdot(\frac{1}{2}-\delta))\cdot(\frac{1}{4}+ \frac{\delta}{3})}\right)}$. \\[4pt]
Piecing everything together,

\begin{itemize}[leftmargin=*]
     \item For the case that $v(i^*) \leq \frac{\beta}{z}\cdot \frac{e-1}{e}\cdot (\frac{1}{2}-\delta)\cdot  v(S^*)$, with probability $1-q$ we run the mechanism below line 4, which gives the factor \[\min\bigg\{
   \frac{(1-q)\,\tilde{p}\,(1-2\delta-2\beta)}{2},\  \frac{(1-q)\,\tilde{p}\,(z-1)\beta(e-1)\left(1 - 2\delta\right)}{2ze}\bigg\},\] where $\tilde{p}=1-2\ \exp{\left(-\frac{\delta^2/ 2 }{(\frac{\beta}{z}\cdot\frac{e-1}{e}\cdot(\frac{1}{2}-\delta))\cdot(\frac{1}{4}+ \frac{\delta}{3})}\right)}$. 

   \item For the case that $v(i^*) > \frac{\beta}{z}\cdot\frac{e-1}{e}\cdot(\frac{1}{2}-\delta)\cdot v(S^*)$,  with probability $q$ we run Dynkin's algorithm , which gives the factor
    \begin{equation}
         \frac{q\,\beta\,(e-1)\,\left(1 - 2\delta\right)}{2ze^{2}}. \tag*{\qedhere}
    \end{equation}
\end{itemize} \end{proof}

\noindent
By combining \Cref{lem:mechanism-with-prediction} and \Cref{lem:mechanism-without-prediction}, the proof of \Cref{thm:approx-ratio-conv-mech} follows.

\section{A Sampling-Calibrated Mechanism with Prediction Guidance}\label{sec:mono-sm-sampling}
In this section, we present a second learning-augmented mechanism for monotone submodular valuations that takes a different approach.
In this mechanism, rather than randomly deciding between using a prediction-based mechanism and one that ignores the prediction, we allow the designer to control the trade-off between prediction reliance and sampling through a parameter $k$, which determines the size of the sample set. The mechanism then combines the information obtained from sampling with the prediction to offer posted prices. When the designer places greater trust in the quality of the prediction, a smaller sample size may suffice, as the reduced accuracy from limited sampling can be compensated for by the predicted information. This mechanism represents a more 'risk-averse' design approach; although it lacks the strong performance guarantees of \Cref{mechanism:mono-sm-conv}, its outcomes are less dependent on randomness. In particular, under a poor prediction, \Cref{mechanism:mono-sm-conv} yields a meaningful guarantee only if the randomized choice favors \Cref{mechanism:mono-sm-withoutpred}. In contrast, \Cref{mechanism:mono-sm2} provides a non-trivial approximation guarantee even when the prediction is entirely uninformative, as it utilizes the information gained through sampling. 

\begin{mechanism}[!ht]
\caption{Online mechanism for a monotone submodular function $v$, parameterized by $k\geq 1$.}\label{mechanism:mono-sm2}
\DontPrintSemicolon
\setstretch{1.1}
\SetInd{4pt}{7pt} 
\begin{algolabel}{}[\textbf{With probability} $q$\,\textbf{:} \SetNoFillComment {\small \tcc*[f]{Dynkin's Algorithm}}]
Sample the first $\floor{n/e}$ agents; let $i^*$ be the most valuable among them.

From the remaining agents, return the first agent $j$ for whom $v(j)\geq v(i^*)$ and pay her $B$.
\end{algolabel}

\begin{algolabel}{}[\textbf{With probability $1-q$\,:}]

Draw $\xi_1$ from the binomial distribution $\mathcal{B}(n,\frac{1}{k})$.

Let $N_1$ be the set of the first $\xi_1$ agents that arrive.  \

Let $T_1$ be a $\frac{e-1}{e}$-approximate solution to \eqref{eq:opt} on $N_1$, e.g., using the algorithm of \citet{sviridenko2004}.  

Set $N_2=N \setminus N_1, B'=B,S=\emptyset$ and $\vec{p}=\vec{0}.$

Set $t= a\cdot \omega + \beta\cdot v(T_1)$

\For{each round  as agent $i \in N_2$ arrives}{

\If{$c_i \leq \bar{p}_i:=\frac{B}{t} v(i \,|\, S) $ and $B'-\bar{p}_i \geq 0$}{

    Add $i$ to $S$, set $p_i=\bar{p}_i$ and $B'=B'-p_i$.
}
}
\end{algolabel}

\Return winning set $S$ and payments $\vec{p}$.
\end{mechanism}\smallskip
\begin{lemma}\label{lem:truth-bf-mono2}
\Cref{mechanism:mono-sm2} is truthful, budget-feasible and individually rational.    
\end{lemma}
\noindent The proof follows the same approach as the proof of \Cref{lem:truth-bf 1}. The following theorem states the approximation guarantee of \Cref{mechanism:mono-sm2}.
\begin{restatable}{theorem}{thmapproxmonosmtwo}\label{thm:approx-mono-sm2}
    \Cref{mechanism:mono-sm2}, given access to a prediction $\omega=(1-\varepsilon)\cdot v(S^*)$ with $\varepsilon\in [0,1)$, has an approximation guarantee of:
    \begin{align*}
        &\min\left\{ \frac{q \left(ek a(1-\varepsilon) + \beta (e-1) \left( 1- k\delta \right) \right)}{k e^2 z},
\right.\\ & \left. (1-q)  \tilde{p}  \min\left(\frac{(k - 1) - k\left(\delta + a(1 - \varepsilon) + \beta\right)}{k}
, \frac{(z - 1)\left( a(1 - \varepsilon)  e k \;+\; \beta (e - 1)(1 - k\delta) \right)}{z e k}
 \right)\right\},
    \end{align*}
    where $\tilde{p}=1-2 \displaystyle\exp{\left(-\frac{3 e k^3 z \, \delta^2}{2 \left(3(k - 1) + \delta k^2\right)\left(e k a(1 - \varepsilon) + \beta (e - 1)(1 - k \delta)\right)}
\right)}$.
\end{restatable}
From the above theorem, we derive the following corollary regarding the in-expectation consistency and robustness guarantees of \Cref{mechanism:mono-sm2}.

\begin{corollary}\label{corolarry:cons-rob-mono-sm2}
    \Cref{mechanism:mono-sm2}, given access to a possibly erroneous prediction $\omega=(1-\varepsilon)\cdot v(S^*)$, with $\varepsilon\in [0,1)$, achieves an in-expectation consistency-robustness tradeoff of $95$ and $280$ by setting $p=0.68$, $a=0.06$, $z=2.15$, $\beta=0.27$, $\delta=0.22$ and $k=2.5$.
\end{corollary}
 
\noindent We present below the proof of \Cref{thm:approx-mono-sm2}, following the same approach as in the proof of \Cref{lem:mechanism-without-prediction}.
\smallskip

\begin{proof}[Proof of Theorem \ref{thm:approx-mono-sm2}]
We distinguish three cases.
\newline
\textbf{Case 1:} \underline{Agent $i^*$ has significant value, i.e.: $v(i^*) >\frac{1}{z}\cdot ( a\cdot (1-\varepsilon)+ \beta\cdot \frac{e-1}{e}\cdot(\frac{1}{k}-\delta))\cdot v(S^*)$.}\\[4pt]
For this case, we will utilize the fact that with probability $q$ we run Dynkin's algorithm, which gives a  factor of $1/e$ over single-agent solutions. Therefore, independently of what our mechanism does with the remaining probability, the expected value for $v(S)$ will be at least:
\begin{align*}
    E[v(S)] &\geq \frac{q}{e\cdot z}\cdot ( a\cdot (1-\varepsilon)+ \beta\cdot \frac{e-1}{e}\cdot(\frac{1}{k}-\delta))\cdot v(S^*) \\&=\frac{q}{ez}(\frac{eka(1-\varepsilon)+\beta(e-1)(1-k\delta)}{ke})=\frac{q(eka(1-\varepsilon)+\beta (e-1)(1-k\delta))}{ke^2z}.
\end{align*}
The following lemma, is the analogous of \Cref{lem:Hoeffding}.
\begin{lemma}\label{lem:Hoeffding2}
When $v(i^*) \leq \frac{1}{z}\cdot ( a\cdot (1-\varepsilon)+\beta\cdot \frac{e-1}{e}\cdot(\frac{1}{k}-\delta))\cdot v(S^*)$, for all $i\in N$ then with probability at least $1-2\ \exp{\Big(-\frac{\delta^2/ 2 }{(\frac{k-1}{k^2}+\frac{\delta}{3})\cdot\frac{1}{z}\cdot ( a\cdot (1-\varepsilon)+\beta\cdot \frac{e-1}{e}\cdot (\frac{1}{k}-
\delta))}\Big)}$, it holds that $(\frac{1}{k}-\delta)\cdot v(S^*)\leq v(S_1^*)$, $(\frac{k-1}{k}-\delta)\cdot v(S^*)\leq v(S_2^*) $, where $\delta $ is a positive constant, that we will set to a certain value later.
\end{lemma}
\begin{proof}[Proof of Lemma \ref{lem:Hoeffding2}]
Fix $S^* = \{a_1,\dots, a_\ell\}$ to be the winning set of an optimal offline solution (arbitrarily chosen if there is more than one). Let $S_1^*=N_1\cap S^*$, $S_2^*=N_2\cap S^*$. We can rearrange the agents in the optimal solution $S^*= \{a_1,\dots, a_\ell\}$, so that they are sorted according to decreasing marginal contributions, i.e., $i\in \argmax_j v(\{a_1,\dots, a_{i-1}\} \cup \{j\}) - v(\{a_1,\dots, a_{i-1}\})$ for $i>1$, and for $i=1$, $v(a_1) = \max_j v(\{a_j\})$.
Let $w_i$ denote the marginal contribution of $a_i$ w.r.t the set $\{a_1,\dots, a_{i-1}\}$.
Since the agents are assumed to arrive in a uniformly random order, we have that each agent is in $N_1$ with probability $\frac{1}{k}$ independently of other agents. Consider the random variables $X_1, \dots, X_\ell$ corresponding to the agents of the optimal solution $S^*$, defined as follows: for $i=1,\dots, \ell$, $X_i=w_i$ when $a_i \in N_1$ and $X_i=0$ otherwise. The previous discussion implies that $X_i$ takes the value $w_i$ with probability $\frac{1}{k}$. Next, we define $X = \sum_{i=1}^\ell X_i$. Symmetrically, consider the random variables $\bar{X_1}, \dots, \bar{X_\ell}$, where  $\bar{X_i}=w_i$ when $a_i \in N_2$ and $\bar{X_i}=0$ otherwise and define $\bar{X}=\sum_{i=1}^\ell \bar{X_i}$.  Observe that $\EX[X]=\frac{1}{k}\cdot v(S^*)$ and due to submodularity, $v(S_2^*) \geq \bar{X}$. Moreover, we have that since the random variables $X_1,\dots,X_{\ell}$ are independent, $\mathrm{Var}[X]=\sum\limits_{i=1}^{\ell}\mathrm{Var}[X_i]=\sum\limits_{i=1}^{\ell}\frac{(k-1)\cdot w_i^2}{k^2}\leq \sum\limits_{i=1}^{\ell}\frac{(k-1)\cdot w_i\cdot v(i^*)}{k^2}=\frac{k-1}{k^2}\cdot\frac{1}{z}\cdot (a\cdot (1-\varepsilon)+\beta\cdot \frac{e-1}{e}\cdot (\frac{1}{k}-
\delta))\cdot v(S^*)^2$. Finally, it holds that $X+\bar{X} = \sum_{i\in [\ell]} w_i=  v(S^*)$.
This implies that $v(S_2^*)+X\geq v(S^*)\Rightarrow v(S_2^*)\geq v(S^*)-X$. 

\noindent Repeating that  $X = \sum\limits_{i=1}^\ell X_i$, with $\mathbb{E}[X]=\frac{1}{k}v(S^*)$ it holds that $v(S_1^*) \geq X$ and $v(S_2^*)\geq v(S^*) - X$,  due to submodularity. Let $m=\frac{X}{v(S^*)}$ with $\EX[m]=\frac{1}{k} $ and $\mathrm{Var}[m]=\frac{\mathrm{Var}[X]}{v(S^*)^2}$. 

Using Bernstein's inequality (\Cref{thm:bernstein}) and the obtained upper bound on the variance of $X,$ we bound the distance of $m$ from its expectation as follows:
\begin{align*}
    &Pr(|m-\EX[m]|\geq \delta)\leq 2 \exp{\left(-\frac{\delta^2/ 2 }{\mathrm{Var}[m]+\frac{1}{z}( a\cdot (1-\varepsilon)+ \beta\cdot\frac{e-1}{e}\cdot(\frac{1}{k}-\delta))\cdot \frac{\delta}{3}}\right)}\\& \leq  2 \exp{\left(-\frac{\delta^2/ 2 }{\frac{k-1}{k^2}\cdot \frac{1}{z}\cdot (a\cdot (1-\varepsilon)+\beta\cdot \frac{e-1}{e}\cdot (\frac{1}{k}-
\delta))+\frac{1}{z}\cdot( a\cdot (1-\varepsilon)+ \beta\cdot\frac{e-1}{e}\cdot(\frac{1}{k}-\delta))\cdot \frac{\delta}{3}} \right)} =\\&=2 \exp{\left(-\frac{\delta^2/ 2 }{(\frac{k-1}{k^2}+\frac{\delta}{3})\cdot \frac{1}{z}\cdot ( a\cdot (1-\varepsilon)+\beta\cdot \frac{e-1}{e}\cdot (\frac{1}{k}-\delta))}\right)}\Rightarrow\\ &Pr(\EX[m]-\delta \leq m\leq \EX[m]+\delta)\geq 1-2 \exp{\left(-\frac{\delta^2/ 2 }{(\frac{k-1}{k^2}+\frac{\delta}{3})\cdot \frac{1}{z}\cdot ( a\cdot (1-\varepsilon)+( \beta\cdot \frac{e-1}{e}\cdot (\frac{1}{k}-\delta))}\right)}.\\
\end{align*}

\noindent Hence, with probability greater or equal to $1-2 \exp{\left(-\frac{\delta^2/ 2 }{(\frac{k-1}{k^2}+\frac{\delta}{3})\cdot\frac{1}{z}\cdot (a\cdot (1-\varepsilon)+ \beta\cdot \frac{e-1}{e}\cdot (\frac{1}{k}-
\delta))}\right)}$, it holds that $(\frac{1}{k}-\delta)\cdot v(S^*)\leq X\leq (\frac{1}{k}+\delta)\cdot v(S^*)$. So, $v(S_1^*)\geq(\frac{1}{k}-\delta)\cdot v(S^*) $ and $v(S_2^*)\geq v(S^*)-(\frac{1}{k}+\delta)\cdot v(S^*)=(\frac{k-1}{k}-\delta)\cdot v(S^*)$, with probability at least $1-2 \exp{\left(-\frac{\delta^2/ 2 }{(\frac{k-1}{k^2}+\frac{\delta}{3})\cdot\frac{1}{z}\cdot (a\cdot (1-\varepsilon)+ \beta\cdot \frac{e-1}{e}\cdot (\frac{1}{k}-\delta))}\right)}$, which completes the proof of the lemma.
\end{proof}

\noindent\textbf{Case 2:} \underline{All $i\in S_2^*\setminus S$ reject the posted price and  $v(i^*) \leq\frac{1}{z}\cdot ( a\cdot (1-\varepsilon)+ \beta\cdot \frac{e-1}{e}\cdot(\frac{1}{k}-\delta))\cdot v(S^*)$.}\\[4pt] Then,
\begin{align*}
        \sum\limits_{i\in S_2^*\setminus S}(v(S_i\cup\{i\})-v(S_i))<\sum\limits_{i\in S_2^*\setminus S}\frac{t}{B}\cdot  c_i = \frac{ a\cdot \omega + \beta\cdot v(T_1)}{B}\cdot B \le ( a\cdot (1-\varepsilon)+ \beta)\cdot  v(S^*),
\end{align*}
where the first inequality holds, due to the fact that each agent in $S_2^*\setminus S$ rejected the posted price offered by the mechanism and the second from the fact that $\omega=(1-\varepsilon)\cdot v(S^*)$ and that the partial solution from $N_1$ is upper bounded by the optimal solution. This inequality leads us to the following one:
\begin{align*}
    v(S_2^*)-v(S)\le ( a\cdot (1-\varepsilon)+\beta) \cdot v(S^*)\Leftrightarrow v(S)\geq v(S_2^*)-( a\cdot (1-\varepsilon)+ \beta) \cdot v(S^*).
\end{align*}
By plugging the bound of \Cref{lem:Hoeffding2} into the above equation and rearranging the terms we get
\[v(S)\geq \left(\frac{k-1}{k}-\delta -( a\cdot (1-\varepsilon)+ \beta)\right) \cdot v(S^*)=\frac{(k-1)-k(\delta+a(1-\varepsilon)+\beta)}{k}\cdot v(S^*)\]
with probability at least $1-2 \exp{\left(-\frac{\delta^2/ 2 }{(\frac{k-1}{k^2}+\frac{\delta}{3})\cdot\frac{1}{z}\cdot (a\cdot (1-\varepsilon)+ \beta\cdot \frac{e-1}{e}\cdot (\frac{1}{k}-
\delta))}\right)}$. \\

\noindent\textbf{Case 3:} \underline{There is an agent $j\in S_2^*\setminus S$ such that $c_j\leq \bar{p}_j$ but $\bar{p}_j >B'$ and also} \\ \underline{ $v(i^*) \leq\frac{1}{z}\cdot ( a\cdot (1-\varepsilon)+ \beta\cdot \frac{e-1}{e}\cdot(\frac{1}{k}-\delta))\cdot v(S^*)$.}\\[4pt]
Let $j\in S_2^*\setminus S$ be the first agent that has this property.

\noindent First we will show that  $B'\leq \frac{B}{z}$. To see this:
    \begin{align*}
        B'< \bar{p}_j =\frac{B}{t}(v(S_j\cup \{j\})-v(S_j))\leq \frac{B}{t} v(j)\leq \frac{B\cdot v(j)}{ a\cdot \omega+\beta\cdot  v(T_1)} \leq \frac{B}{z},
    \end{align*}
 where the last inequality holds due to the fact that $v(j)\leq v(i^*)\leq \frac{1}{z} (a\cdot (1-\varepsilon)+ \beta\cdot \frac{e-1}{e}\cdot(\frac{1}{k}-\delta))\cdot v(S^*)\leq \frac{a\cdot \omega+ \beta \cdot v(T_1)}{z}$, given the event of \Cref{lem:Hoeffding2}. We have also used submodularity, since $v(j|S_j)\leq v(j)\leq v(i^*)$. 

\noindent Since $B'\leq \frac{B}{z}$, we have that, $\sum\limits_{i\in S} p_i\geq B-\frac{B}{z}=\frac{(z-1)}{z}\cdot B$ and also:
    \begin{align*}
        \frac{z-1}{z}\cdot B\leq\sum\limits_{i\in S} p_i=\sum\limits_{i\in S}\frac{B}{t}(v(S_i\cup \{i\})-v(S_i))\leq B\cdot \frac{v(S)}{ a\cdot \omega+\beta \cdot v(T_1)},
    \end{align*}
    which leads us to $v(S)\geq \frac{z-1}{z}\cdot( a\cdot \omega+ \beta \cdot v(T_1))\geq \frac{z-1}{z}\cdot ( a\cdot (1-\varepsilon)\cdot v(S^*)+ \beta \cdot \frac{e-1}{e}\cdot  v(S_1^*))$.  The latter, via \Cref{lem:Hoeffding2}, gives us that: 
    \begin{align*}        
    v(S)&\geq  \frac{z-1}{z}\cdot ( a\cdot (1-\varepsilon)+ \beta \cdot \frac{e-1}{e}\cdot (\frac{1}{k}-\delta) )\cdot v(S^*)\\&=\frac{(z-1)(a(1-\varepsilon)ek+\beta(e-1)(1-k\delta))}{zek}\cdot v(S^*), 
    \end{align*}
    with probability at least $1-2\exp{\left(-\frac{\delta^2/ 2 }{(\frac{k-1}{k^2}+\frac{\delta}{3})\cdot\frac{1}{z}\cdot (a\cdot (1-\varepsilon)+ \beta\cdot \frac{e-1}{e}\cdot (\frac{1}{k}- \delta))}\right)}$. \\
Piecing everything together we get the following,

\begin{itemize}[leftmargin=*]
       
    \item For the case that $v(i^*) \leq \frac{1}{z}\cdot (a\cdot (1-\varepsilon)+ \beta\cdot \frac{e-1}{e}\cdot(\frac{1}{k}-\delta))\cdot v(S^*)$, with probability $1-q$ we run the mechanism below line 4, which gives the factor \[(1-q) \cdot \tilde{p}\cdot  \min{\left(\frac{(k-1)-k(\delta+a(1-\varepsilon)+\beta)}{k}, \frac{(z-1)(a(1-\varepsilon)ek+\beta(e-1)(1-k\delta))}{zek} \right)},\]  
    where $\tilde{p}=1-2 \exp{\left(-\frac{\delta^2/ 2 }{(\frac{k-1}{k^2}+\frac{\delta}{3})\cdot\frac{1}{z}\cdot (a\cdot (1-\varepsilon)+ \beta\cdot \frac{e-1}{e}\cdot (\frac{1}{k}-\delta))}\right)}$.  
    \item For the case that $v(i^*) \ge\frac{1}{z}\cdot ( a\cdot (1-\varepsilon)+ \beta\cdot \frac{e-1}{e}\cdot(\frac{1}{k}-\delta))\cdot v(S^*)$ , with probability $q$ we run Dynkin's algorithm , which gives the factor
    \begin{equation}
        \frac{q(eka(1-\varepsilon)+\beta (e-1)(1-k\delta))}{ke^2z}. \tag*{\qedhere}
    \end{equation} 
\end{itemize} 
This concludes the proof of the theorem.
\end{proof}

\section{Mechanisms for the Non-Monotone Case}\label{sec:nonmono-submodular}
In this section we present randomised mechanisms for a general (possibly non-monotone) submodular utility function, in the secretary model of arrival. These mechanisms follow the same approach as \Cref{mechanism:mono-sm-conv,mechanism:mono-sm2} for the monotone case.   The difference is that for non-monotone submodular functions, it does not suffice to build a single solution in order to get any meaningful approximation guarantee. We circumvent this difficulty, by building two disjoint solutions simultaneously, which is an idea firstly introduced in \citet{amanatidis2022budget}. Since we are in the Random Order Model, any decision the mechanism makes upon the arrival of an agent must be irrevocable. Therefore,  we are not allowed to return the best of the two solutions built and instead, the mechanism selects equiprobably, which of the two will be returned \emph{beforehand}. For the performance analysis of the mechanism, we will follow the same reasoning as in \Cref{sec:mono-submodular}. Missing proofs from this section are in \ref{app:sec5}.

\subsection{The Non-Monotone Analogue of \Cref{mechanism:mono-sm-conv}}
Here we will present the analogue of \Cref{mechanism:mono-sm-conv}, for the non-monotone case. Again, this is a convex combination of two other budget-feasible mechanisms, where the first one builds a solution based solely on the information of the prediction $\omega$ and the second completely ignores the prediction and builds the solution based on information attained by sampling roughly the first $\frac{n}{2}$ agents that arrive. 
\begin{theorem} \label{thm:approx-ratio-nonmono-conv-mech}
    \Cref{mechanism:nonmono-sm-conv}  with access to a prediction $\omega=(1-\varepsilon)\cdot v(S^*)$ with $\varepsilon\in [0,1)$, achieves an in-expectation approximation guarantee of $(\tau \cdot f_1+(1-\tau)\cdot f_2)^{-1}$, where $f_1$  is the expected approximation ratio of \Cref{mechanism:nonmono-sm-pred} and is equal to 
    \[\min\left\{\frac{ap \; (1-\varepsilon)\;}{\gamma},\frac{(1-p)\; (1-(1-\varepsilon)\; 2a)}{4},\frac{(1-p)\; a\; (1-\varepsilon)\;(\gamma-1)}{2 \gamma}\right\} \]
    and $f_2$ is the expected approximation ratio of \Cref{mechanism:nonmono-sm-withoutpred}, which is 
    \[\min\left\{\frac{q \beta (1 - 2\delta)}{2 e^2z }
, \frac{(1 - q)\tilde{p} (1 - 2\delta - 4\beta)}{8}
, \frac{(1 - q)\tilde{p} \beta (z - 1)(1 - 2\delta)}{4 e z}
\right\}\,,\]
    with $\tilde{p}=1-2\cdot \displaystyle\exp{\left(-\frac{12 z e \, \delta^2}{\beta (1 - 2\delta)(3 + 4\delta)}
\right)}.$ 
\end{theorem}
From \Cref{thm:approx-ratio-nonmono-conv-mech}, we obtain the following two corollaries, which give us the in-expectation consistency and robustness guarantees of \Cref{mechanism:nonmono-sm-conv}.
\begin{corollary}\label{corolarry:consistency-nonmono-conv-mech}
    \Cref{mechanism:nonmono-sm-conv}, given access to a perfect prediction $\omega=v(S^*)$ achieves an in-expectation consistency factor of $(\tau\cdot \frac{1}{19}+(1-\tau)\cdot \frac{1}{445})^{-1}$, by setting $p=0.33$, $a=0.335$, $\gamma=2$ for \Cref{mechanism:nonmono-sm-pred} $q=0.63$, $z=2.395$, $\beta=0.171$, $\delta=0.13$ for \Cref{mechanism:nonmono-sm-withoutpred}.
\end{corollary}
\begin{corollary}
    \Cref{mechanism:nonmono-sm-conv}, given access to an erroneous prediction $\omega=(1-\varepsilon)\cdot v(S^*)$, with $\varepsilon\in (0,1)$, achieves an in-expectation robustness factor of $(\tau\cdot \frac{1}{19}\cdot (1-\varepsilon)+(1-\tau)\cdot \frac{1}{445})^{-1}$ by setting $p=0.33$, $a=0.335$, $\gamma=2$ for \Cref{mechanism:nonmono-sm-pred} and $q=0.63$, $z=2.395$, $\beta=0.171$, $\delta=0.13$ for \Cref{mechanism:nonmono-sm-withoutpred}. 
\end{corollary}
\noindent The proof of \Cref{thm:approx-ratio-nonmono-conv-mech}, consists of \Cref{lem:mechanism-nonmono-with-prediction} and \Cref{lem:mechanism-nonmono-without-prediction}, but first we will discuss truthfulness, budget feasibility and individual rationality.

\setlength{\algomargin}{9pt}
\makeatletter
\patchcmd{\@algocf@start}
  {-1.5em}
  {2pt}
{}{}
\makeatother


\begin{mechanism}[!h]
\caption{Online mechanism for a possibly non-monotone submodular function $v$, parameterized by tolerance $\tau \in [0,1)$.
}\label{mechanism:nonmono-sm-conv}
\DontPrintSemicolon
\setstretch{1.1}
\SetInd{4pt}{7pt} 
\begin{algolabel}{}[\textbf{With probability} $\tau$\,\textbf{:}]
Run \Cref{mechanism:nonmono-sm-pred} \label{line:Mechanism-pred2}

\end{algolabel}

\begin{algolabel}{}[\textbf{With probability $1-\tau$\,:}]

Run \Cref{mechanism:nonmono-sm-withoutpred} \label{line:Mechanism-withoutpred2}
\end{algolabel}

\end{mechanism}\smallskip

\setlength{\algomargin}{9pt}
\makeatletter
\patchcmd{\@algocf@start}
  {-1.5em}
  {2pt}
{}{}
\makeatother


\begin{mechanism}[!ht]
\caption{Online mechanism for a possibly non-monotone submodular function $v$, augmented with an (underestimating) prediction $\omega$ of $v(S^*)$, which it uses to build a feasible solution. 
}\label{mechanism:nonmono-sm-pred}
\DontPrintSemicolon
\setstretch{1.1}
\SetInd{4pt}{7pt} 
\begin{algolabel}{}[\textbf{With probability} $p$\,\textbf{:}]

Return the first agent $j$ for whom $v(j)\geq \frac{a\cdot \omega}{\gamma}$ and pay her $B$. \label{line:Dynkin2}
\end{algolabel}

\begin{algolabel}{}
[\textbf{With probability $1-p$:}]
Set $t=a\cdot\omega$\\
Set $S=S_1=S_2=\emptyset;$ $B_1=B_2=B$.

$S= S_j$, where $j$ is uniformly drawn from $\{1,2\}$.

\For{each round  as an agent $i$ arrives}{
Let $j(i)=\argmax_{j\in\{1,2\}}v(i|S_j)$

\If{$c_i \leq \bar{p}_i:=\frac{B}{t} v(i \,|\, S_{j(i)}) $ and $B_{j(i)}-\bar{p}_i \geq 0$}{

Add $i$ to $S_{j(i)}$, set $p_i=\bar{p}_i$ and $B_{j(i)}=B_{j(i)}-\bar{p}_i$.
    
}
}

Set $p_i = 0$ for any $i\notin S$.
\end{algolabel}
\Return set $S$ and payments $\vec{p}$.

\end{mechanism}\smallskip

\setlength{\algomargin}{9pt}
\makeatletter
\patchcmd{\@algocf@start}
  {-1.5em}
  {2pt}
{}{}
\makeatother


\LinesNotNumbered

\begin{mechanism}[!ht]
\LinesNotNumbered 
\caption{Online mechanism for a possibly non-monotone submodular function $v$.}\label{mechanism:nonmono-sm-withoutpred}
\DontPrintSemicolon
 \LinesNotNumbered 
\SetKwProg{Dots}{\ldots}{ \ldots}{\ldots}
\setstretch{1.1}
\SetInd{4pt}{7pt}
\SetNlSty{textbf}{}{} 
\Dots{(Lines 1--8 same as in \Cref{mechanism:mono-sm-withoutpred}, but now $T_1$ is a $1/e$-approximate solution on $N_1$, using the algorithm of \citet{kulik2013}, and $q$,  $\beta$ have different values.)}

\nlset{9} Set $S=S_1=S_2=\emptyset;$ $B_1=B_2=B$.

\nlset{10}$S= S_j$, where $j$ is uniformly drawn from $\{1,2\}$.

\nlset{11}\For{each round as $i \in N_2$ arrives}{
\nlset{12}Let $j(i)=\argmax_{j\in\{1,2\}}v(i|S_j)$

\nlset{13}\If{$c_i \leq \bar{p}_i:=\frac{B}{t} v(i \,|\, S_{j(i)}) $ and $B_{j(i)}-\bar{p}_i \geq 0$} {

  \nlset{14} Add $i$ to $S_{j(i)}$, set $p_i=\bar{p}_i$ and $B_{j(i)}=B_{j(i)}-\bar{p}_i$.
    
}
}

\nlset{15}Set $p_i = 0$ for any $i\notin S$.

\nlset{16}\Return set $S$ and payments $\vec{p}$.
\end{mechanism}
\LinesNumbered
\begin{restatable}{lemma}{lemtruthbfnonmono}\label{lem:truth-bf-nonmono}
\Cref{mechanism:nonmono-sm-conv} is truthful, budget-feasible and individually rational.    
\end{restatable}

\begin{lemma}\label{lem:mechanism-nonmono-with-prediction}
    \Cref{mechanism:nonmono-sm-pred}, with access to a prediction $\omega=(1-\varepsilon)\cdot v(S^*)$ with $\varepsilon\in [0,1)$,  has an approximation ratio of $\min\!\left\{
    \frac{a p(1-\varepsilon) }{\gamma},\frac{(1-p)(1-(1-\varepsilon)2a)}{4}, \frac{(1-p)a(1-\varepsilon)(\gamma-1)}{2\gamma}\right\}$ in expectation.
\end{lemma}
\begin{proof}[Proof Sketch]
    The proof is quite similar to the one of \Cref{lem:mechanism-with-prediction}, it discriminates between the same three cases and proceeds in a similar way. In order to avoid repeating the same arguments, we will restrict this proof to the second case, which requires different calculations due to the non-monotonicity of the value functions. Also, we will use some extra notation, for an agent $i$ we denote as $S_{j(i)}$, $j(i)\in \{1,2\}$ the set in which agent $i$ had the greatest marginal contribution when arrived and with $S_j^{t_i}$ the set $S_j$ at the moment $t_i$ that agent $i$ arrived.
\begin{itemize}
        \item \underline{Each agent $i\in S^*\setminus (S_1\cup S_2)$ rejects the posted price and $v(i^*)\leq \frac{a\cdot (1-\varepsilon)}{\gamma}\cdot v(S^*)$.}\\[4pt] 
From the definition of submodular functions we have that:
\begin{align*}
    v(S_1\cup S^*)+v(S_2\cup S^*)\geq v(S^*)+v(S_1\cup S_2\cup S^*)\geq v(S^*),
\end{align*}where the last inequality holds because the value function $v$ is non-negative. From the definition of submodularity again we get:
\begin{align*}
    v(S_1\cup S^*) &\leq v(S_1)+\sum\limits_{i\in S^*\setminus S_1}v(i|S_1)\\&=v(S_1)+\sum\limits_{i\in S^*\setminus(S_1\cup S_2)} v(i|S_1)+\sum\limits_{i\in (S^*\cap S_2)\setminus S_1} v(i|S_1)\\ &\leq v(S_1)+\sum\limits_{i\in S^*\setminus(S_1\cup S_2)} v(i|S_1^{t_i})+\sum\limits_{i\in (S^*\cap S_2)\setminus S_1} v(i|S_1^{t_i})\\&\leq v(S_1)+\sum\limits_{i\in S^*\setminus(S_1\cup S_2)} v(i|S_1^{t_i})+\sum\limits_{i\in (S^*\cap S_2)\setminus S_1} v(i|S_{2}^{t_i}) \\ &\leq v(S_1)+\sum\limits_{i\in S^*\setminus(S_1\cup S_2)} v(i|S_1^{t_i})+\sum\limits_{i\in S_2} v(i|S_{2}^{t_i}) 
    \\&= v(S_1)+v(S_2)+\sum\limits_{i\in S^*\setminus(S_1\cup S_2)} v(i|S_1^{t_i}),
\end{align*}
where the first inequality holds because $S_1^{t_i}\subseteq S_1$ and the second one due to the fact that each agent  $i\in S^*\cap S_2$ was added to set $S_2$, hence $v(i|S_1^{t_i})\leq v(i|S_2^{t_i})$. Then, the fourth inequality holds because
$(S^*\cap S_2)\setminus S_1\subseteq S_2$ and every agent $i\in S_2$ has a strictly
positive marginal $v(i|S_2^{t_i})>0$: indeed, $i$ was accepted into $S_2$, so it
passed the threshold $c_i\leq \frac{B}{t}\,v(i|S_2^{t_i})$, which together with
$c_i>0$ gives $v(i|S_2^{t_i})>0$. Symmetrically, we get that: 
\begin{align*}
    v(S_2\cup S^*)\leq v(S_2)+v(S_1)+\sum\limits_{i\in S^*\setminus(S_1\cup S_2)} v(i|S_2^{t_i}).
\end{align*}
Adding up, these two inequalities and dividing by 4, we have that: 
\begin{align*}
\frac{v(S_1)}{2}+\frac{v(S_2)}{2}\geq \frac{v(S^*)}{4}-\frac{1}{2} \sum\limits_{i\in S^*\setminus(S_1\cup S_2)} v(i|S_{j(i)}^{t_i}),
\end{align*} which combined with the fact that each agent $i\in S^*\setminus(S_1\cup S_2)$ rejected the posted price offered by the mechanism and also that $\mathbb{E}[v(S)]=\frac{v(S_1)}{2}+\frac{v(S_2)}{2}$, leads to: 
\begin{align*}
\EX[v(S)]&\geq \frac{v(S^*)}{4}-\frac{\sum\limits_{i\in S^*\setminus(S_1\cup S_2)} c_i \cdot a\cdot(1-\varepsilon)\cdot v(S^*)}{2\cdot B}
\\ &\geq  \frac{v(S^*)}{4}-\frac{ a\cdot(1-\varepsilon)\cdot v(S^*)}{2}=\frac{1-2a\cdot(1-\varepsilon)}{4}\cdot v(S^*),
\end{align*}
where the last inequality holds due to budget feasibility of set $S^*$. \qedhere
\end{itemize}
\end{proof}

    \begin{restatable}{lemma}{lemmechanismnonmonowithoutprediction}\label{lem:mechanism-nonmono-without-prediction}
    \Cref{mechanism:nonmono-sm-withoutpred}, has an approximation ratio of:
    \begin{align*} 
\min\left\{ 
(1 - q)\tilde{p} \cdot \min\left( 
\frac{1 - 2\delta - 4\beta}{8},\ 
\frac{\beta (z - 1)(1 - 2\delta)}{4 e z}
\right),\ 
\frac{q \beta (1 - 2\delta)}{2 e^2 z}
\right\}
     \end{align*}
     in expectation, where $\tilde{p}=1-2 \exp{\left(-\frac{12 z e \, \delta^2}{\beta (1 - 2\delta)(3 + 4\delta)}
\right)}.$   
\end{restatable}

\subsection{The Non-Monotone Analogue of  \Cref{mechanism:mono-sm2} }
Here we will present a mechanism that follows the same approach as \Cref{mechanism:mono-sm2}, adapted for the non-monotone case, by building two solutions instead of a single one. Again, we cannot simply return the best of the two solutions built, since this would be incompatible with the Random Order Model, so the mechanism randomly selects beforehand which of the two solutions to return, each chosen with probability ${1}/{2}$. 
\begin{mechanism}[!ht]
\caption{Online mechanism for a general submodular function $v$, parameterized by tolerance $\tau \in [0,1)$.}\label{mechanism:non-mono-sm2}
\DontPrintSemicolon
\setstretch{1.1}
\SetInd{4pt}{7pt} 
\begin{algolabel}{}[\textbf{With probability} $q$\,\textbf{:}  \SetNoFillComment {\small \tcc*[f]{Dynkin's Algorithm}}]

Sample the first $\floor{n/e}$ agents; let $i^*$ be the most valuable among them.

From the remaining agents, return the first agent $j$ for whom $v(j)\geq v(i^*)$ and pay her $B$.
\end{algolabel}

\begin{algolabel}{}[\textbf{With probability $1-q$\,:}]

Draw $\xi_1$ from the binomial distribution $\mathcal{B}(n,\frac{1}{k})$.

Let $N_1$ be the set of the first $\xi_1$ agents that arrive.  \

Let $T_1$ be a  $(\frac{1}{e})$-approximate solution to \eqref{eq:opt} on $N_1$, using the algorithm of \citet{kulik2013}.

Set $t= a\cdot \omega +\beta\cdot v(T_1)$

Set $S=S_1=S_2=\emptyset$, $B_1=B_2=B$.

$S= S_j$, where $j$ is uniformly drawn from $\{1,2\}$.

\For{each round as $i \in N_2$ arrives}{
Let $j(i)=\argmax_{j\in\{1,2\}}v(i|S_j)$

\If{$c_i \leq \bar{p}_i:=\frac{B}{t} v(i \,|\, S_{j(i)}) $ and $B_{j(i)}-\bar{p}_i \geq 0$}{

   Add $i$ to $S_{j(i)}$, set $p_i=\bar{p}_i$ and $B_{j(i)}=B_{j(i)}-\bar{p}_i$.
    
}
}

Set $p_i = 0$ for any $i\notin S$.
\end{algolabel}
\nlset{15}\Return set $S$ and payments $\vec{p}$.
\end{mechanism}\smallskip
\begin{lemma}\label{lem:truth-bf-nonmono-2}
\Cref{mechanism:non-mono-sm2} is truthful, budget-feasible and individually rational.    
\end{lemma}
The proof of this lemma is similar to the proof of \Cref{lem:truth-bf-nonmono}.
\begin{restatable}{theorem}{thmapproxnonmonosmtwo}\label{thm:approx-non-mono-sm2}
    \Cref{mechanism:non-mono-sm2}, given access to a prediction $\omega=(1-\varepsilon)\cdot v(S^*)$ with $\varepsilon\in[0,1)$, has an approximation guarantee of 
    \begin{align*}
        \min\Bigg\{&
\frac{q \left(e k\,a(1 - \varepsilon) + \beta (1 - k \delta)\right)}{e^2 z k},\\& 
(1 - q)\tilde{p} \cdot \min\Bigg(
\frac{(k - 1 - k\delta) - 2k\left(a(1 - \varepsilon) + \beta\right)}{4k},\ 
\frac{(z - 1)\left(e k\,a(1 - \varepsilon) + \beta (1 - k\delta)\right)}{2 e z k}
\Bigg)
\Bigg\},
    \end{align*}
    where $\tilde{p}=1-2 \displaystyle\exp{\left(-\frac{3 e k^3 z \, \delta^2}{2\left(3(k - 1) + \delta k^2\right)\left(e k a(1 - \varepsilon) + \beta(1 - k\delta)\right)}
\right)}$.
\end{restatable}
From the above theorem, we obtain the following corollary, which gives us the in-expectation consistency and robustness tradeoff of \Cref{mechanism:non-mono-sm2}.
\begin{corollary}\label{corolarry:cons-rob-non-mono-sm2}
    \Cref{mechanism:non-mono-sm2}, given access to a possibly erroneous prediction $\omega=(1-\varepsilon)\cdot v(S^*)$, with $\varepsilon\in [0,1)$, achieves an in-expectation consistency-robustness tradeoff of $228$ and $818$, by setting $q=0.61$, $z=2.47$, $\beta=0.155$, $a=0.035$, $\delta=0.164$ and $k=2.5$.
\end{corollary}

\section{An Impossibility Result}\label{sec:lower-bound}
Here we turn our attention to the \textit{offline} setting and present a lower bound of $2$ on the consistency guarantee for all universally truthful, budget-feasible, and individually rational randomized mechanisms \textit{with predictions}. Clearly, this impossibility result extends to the online setting as well.  Moreover, our result holds even when the instances are augmented with a prediction of the optimal set of agents $\hat{S} \subseteq N$, rather than merely the value of the optimal solution $\omega$ as assumed in \Cref{sec:mono-submodular,sec:nonmono-submodular}. Note that having access to $\hat{S}$ in the offline setting is at least as informative as having access to $\omega$, since the latter can be inferred from the former, i.e., given a prediction $\hat{S}$, we may set $\omega = v(\hat{S})$. We denote by $\mathcal{M}$ the class of universally truthful, budget-feasible, and individually rational randomized mechanisms, and by $\mathcal{M}^d \subseteq \mathcal{M}$ the subclass of deterministic mechanisms that are truthful, budget-feasible, and individually rational.
The main result of this section is \Cref{thm:randomized-lb-2}.

\begin{theorem}\label{thm:randomized-lb-2}
    For any $\varepsilon > 0$, no universally truthful, individually rational,  budget-feasible randomized mechanism that has access to a prediction of the optimal solution can be $(2-\varepsilon)$-consistent, even for  additive valuation functions.
\end{theorem}

To prove \Cref{thm:randomized-lb-2}, we focus on a family of augmented instances with just two agents that differ only in their cost profiles. Specifically, for each such instance, we set $v(\{1\}) = v(\{2\}) = 1$, denote the budget as $B$, and assume access to a prediction $\hat{S}=\{1,2\}$. We refer to these instances as \emph{canonical}. Note that the prediction is perfect only for cost profiles where the optimal solution coincides with it, i.e., when, given a profile $\vec{c} = (c_1, c_2)$, it holds that $c_1 + c_2 \leq B$.

One of the technical tools we use to prove \Cref{thm:randomized-lb-2} is Yao's Lemma (due to \citet{yao1977}). Intuitively, this idea enables us to show a lower bound on the performance of a universally truthful, individually rational and budget-feasible randomized mechanism $M \in \mathcal{M}$ by focusing on the performance of an (arbitrarily given) deterministic truthful, individually rational and budget-feasible mechanism $(A, \vec{p}) \in \mathcal{M}^d$ instead. This is achieved by constructing an ``adversarial'' distribution of instances of canonical instances and reasoning about the performance of $(A, \vec{p})$ for the distribution. \Cref{lemma:yao} is a weaker version of Yao's Lemma for our procurement auction setting which is sufficient for our purposes. For completeness, we provide its proof in \ref{app:sec6}.

\begin{restatable}{lemma}{lemmayao}\label{lemma:yao}
    Let $C$ be a set of cost profiles for canonical instances. For any mechanism $M=(A^r, \vec{p}^r) \in \mathcal{M}$ and any discrete probability distribution P of profiles in $C$, we have
\begin{equation*}\label{eq:yao-minmax}
        \min_{\vec{c} \in C}\EX\left[\frac{v\left(A^r(\vec{c}) \right)}{v\left(S^*(\vec{c}) \right)} \right] \leq  \max_{{(A,\vec{p}) \in
        \mathcal{M}^d}} \EX_{\vec{c} \sim P}\left[\frac{v\left(A(\vec{c}) \right)}{v\left(S^*(\vec{c}) \right)} \right].
    \end{equation*}
\end{restatable}

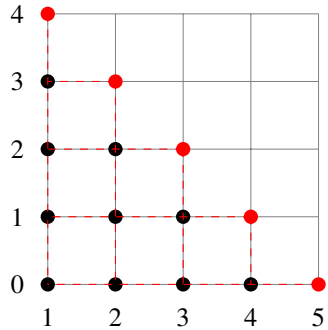
\begin{figure}
    \centering
    \begin{tikzpicture}[scale=0.9]

    \draw[step=1cm, gray, very thin] (0,0) grid (4,4);

    \foreach \x [evaluate=\x as \y using int(\x-1)] in {1,2,3,4,5} {
        \node[below] at ({\x-1},-0.2) {\small\x};
        \node[left] at (-0.2,{\x-1}) {\small\y};
    }

    \fill[black] (0,0) circle (3pt); 
    \fill[black] (0,1) circle (3pt); 
    \fill[black] (0,2) circle (3pt); 
    \fill[black] (0,3) circle (3pt); 
    \fill[black] (1,0) circle (3pt); 
    \fill[black] (1,1) circle (3pt); 
    \fill[black] (1,2) circle (3pt); 
    \fill[black] (2,0) circle (3pt); 
    \fill[black] (2,1) circle (3pt); 
    \fill[black] (3,0) circle (3pt); 

    \fill[red] (4,0) circle (3pt); 
    \fill[red] (3,1) circle (3pt); 
    \fill[red] (2,2) circle (3pt); 
    \fill[red] (1,3) circle (3pt); 
    \fill[red] (0,4) circle (3pt); 

    \draw[red, dashed] (4,0) -- (4,0); 
    \draw[red, dashed] (3,1) -- (3,0); 
    \draw[red, dashed] (2,2) -- (2,0); 
    \draw[red, dashed] (1,3) -- (1,0); 
    \draw[red, dashed] (0,4) -- (0,0); 

    \draw[red, dashed] (3,1) -- (0,1); 
    \draw[red, dashed] (2,2) -- (0,2); 
    \draw[red, dashed] (1,3) -- (0,3); 
\draw[red, dashed] (4,0) -- (0,0);
    \end{tikzpicture}
    \caption{Illustration of the support of $P(5)$ (red dots) for $B = 5$.\label{fig:lw-illustration}}
\end{figure}

We now describe the probability distribution of cost profiles for canonical instances we use. Let $k \geq 2$ be an integer. We denote the distribution by $P(k)$ and draw $\vec{c} \sim P(k)$ as follows:
for $i \in \{1, \dots, k\}$, profile $\vec{c} = \left(iB/k, (k - i)B/k \right)$ is drawn with probability $1/k$; see \Cref{fig:lw-illustration} for an illustration.
Observe that for each such profile the optimal solution to the algorithmic problem is $S^*(\vec{c})= \{1,2\}$ and, thus, the prediction $\hat{S} = \{1,2\}$ is perfect.

Our setting in \Cref{sec:lower-bound} is an offline, single-parameter domain. Therefore, the characterization of \citet{myerson81} applies. In \Cref{lemma:myerson}, we restate a version of the characterization for deterministic mechanisms by \citet{archer01}, which will be sufficient for our purposes.

\begin{lemma}[due to \protect\citet{myerson81}, \protect\citet{archer01}]
\label{lemma:myerson}
A single-parameter deterministic mechanism $M^d = (A, \vec{p})$ is truthful and individually rational if and only if, for every instance $I = (N, \vec{c}, v, B)$ and every agent $i \in N$, the following conditions hold:
\begin{enumerate}
    \item If $i \in A(\vec{c})$, then for all $c'_i \leq c_i$, it holds that $ i \in A(c'_i, \vec{c}_{-i})$ (monotonicity).
    \item If $i \in A(\vec{c})$, then $p_i(\vec{c}) = \sup \left\{ z \geq c_i \mid i \in  A(z, \vec{c}_{-i}) \right\}$. Otherwise, $p_i(\vec{c}) = 0$ (payment identity).
\end{enumerate}
\end{lemma}

The next lemma  describes a necessary trade-off that every truthful and budget-feasible deterministic mechanism $M^d$ must make for canonical instances.

\begin{lemma}\label{lemma:profile-blocking}
    Fix $c \in (0, B]$ and $c' \in [0, c)$. Consider two canonical instances with cost profiles $\left(c, B-c \right)$ and $\left( c', B-c' \right)$. For every mechanism $(A, \vec{p}) \in \mathcal{M}^d$ it holds that $v\left(A\left( c, B-c \right)\right) + v\left(A\left(c', B-c' \right)\right) \leq 3$.
\end{lemma}

\begin{proof}
Toward a contradiction, suppose there exists a mechanism $(A, \vec{p}) \in \mathcal{M}^d$ for which $v(A( c, \allowbreak B-c ))$ $+ v\left(A\left(c', B-c' \right)\right) = 4$. This implies that $A\left(c, B-c \right) = A\left(c', B-c' \right) = \{1, 2\}$. Consider the profile $(c', B-c)$ and note that $c' < c$. By the monotonicity of $(A, \vec{p})$ (property 1 of \Cref{lemma:myerson}), and since $1 \in A(c, B-c)$ holds by assumption, it must be that $1 \in A(c', B-c)$. By the same argument, $2 \in A(c', B-c)$ holds, since $2 \in A(c', B-c')$ is true. We thus conclude that $A(c', B-c)=\{1,2\}$. Moreover,
\begin{align*}
      p_1 (c', B-c)  + p_2(c', B-c)  & =  \sup\left\{ z \geq c' \mid 1 \in A(z,B-c) \right\} +    
       \sup\left\{ z \geq B - c \mid  2 \in A(c',z) \right\} \\
      & \geq c + B-c' >B.
\end{align*}
The first equality follows from the payment identity of \Cref{lemma:myerson}. However, this is a contradiction as $(A, \vec{p})$ is, by assumption, budget-feasible. The lemma follows.
\end{proof}

As we prove in \Cref{lemma:profile-blocking}, no deterministic mechanism in $\mathcal{M}^d$ can select \textit{both} agents in \textit{two} red points, as it would be forced to select both agents in their common projection (the corresponding black point) by Myerson's characterization (\Cref{lemma:myerson}) and violate budget feasibility.
In \Cref{lemma:ratio-P-ub}, we derive an upper bound on the expected approximation ratio of the distribution $P(k)$ for every parameter $k$ and every deterministic mechanism in $\mathcal{M}^d$ which is truthful and budget-feasible.

\begin{restatable}{lemma}{lemmaratiopub}\label{lemma:ratio-P-ub}
Fix an integer $k \geq 2$. For every deterministic mechanism $(A, \vec{p}) \in \mathcal{M}^d$ it holds that
\begin{equation*}
    \EX_{\vec{c} \sim P(k)}\left[\frac{v\left(A(\vec{c})\right)}{v\left(S^*(\vec{c})\right)}\right] \leq \frac{1}{2} \cdot \left(1+\frac{1}{k} \right).
\end{equation*}
\end{restatable}
\begin{proof}
Let $Q$ be the profiles in the support of $P(k)$. Observe that, by definition, for every 
$(iB/k, (k - i)B/k) \in Q$,
it holds that $iB/k + (k - i)B/k = B$.
We can apply \Cref{lemma:profile-blocking} for any two different profiles $(c_1, c_2) \in Q$ and $(c'_1, c'_2) \in Q$, obtaining that $v\left(A\left( c_1, c_2 \right)\right) + v\left(A\left(c'_1, c'_2 \right)\right) \leq 3$. Thus, there can be \emph{at most} one profile $\vec{c}^\star \in Q$ with $v\left(A\left( \vec{c}^\star \right)\right) = 2$. Furthermore, for the remaining profiles $\vec{c} \in Q \setminus \{\vec{c}^\star\}$, it must be that $v\left(A\left( \vec{c}\right)\right) \le 1$.
We conclude that
    \begin{align*}
        \EX_{\vec{c} \sim P(k)}\left[\frac{v\left(A(\vec{c})\right)}{v\left(S^*(\vec{c})\right)}\right] &\leq \frac{1}{k} \cdot \frac{2}{2} + \sum_{\vec{c} \in Q \setminus \{\vec{c}^\star\}}\frac{1}{k} \cdot \frac{1}{2}
        =\frac{1}{2} + \frac{1}{2k},
    \end{align*}
where the equality follows since $|Q \setminus \{\vec{c}^*\}| = k-1$ by definition.
\end{proof}

We can now prove \Cref{thm:randomized-lb-2}.

\begin{proofof}{Proof of \Cref{thm:randomized-lb-2}}
Toward a contradiction, suppose that there exists some $\varepsilon >0$ and a $(2-\varepsilon)$-consistent randomized mechanism $M =(A^r, \vec{p}^r) \in \mathcal{M}$. Let $k= \ceil*{{2}/{\varepsilon}}$.

Recall that our canonical instances are augmented by a prediction $\hat{S} = \{1,2\}$ of the optimal solution. We use $C(\hat{S})$ to refer to the class of canonical instances for which the prediction is perfect. Formally, $C(\hat{S}) = \{\vec{c} \in [0,B]^2 \mid \hat{S} = S^*(\vec{c})\}$. By the assumed consistency guarantee of $M$, we get
\begin{equation*}
    \forall \vec{c} \in C(\hat{S}), \quad \EX\left[v(A^r(\vec{c})) \right] \geq \frac{1}{2 - \varepsilon} \cdot v\left(S^*(\vec{c}) \right).
\end{equation*}
Therefore, we obtain 
\begin{equation}\label{eq:lb-proof-conclusion}
    \frac{1}{2-\varepsilon} \leq \!\min_{\vec{c} \in C(\hat{S})} \frac{\EX\left[v(A^r(\vec{c})) \right]}{v(S^*(\vec{c}))} = \!\min_{\vec{c} \in C(\hat{S})} \!\EX\left[\frac{v(A^r(\vec{c}))}{v(S^*(\vec{c}))} \right] 
    \leq \max_{(A,\vec{p}) \in \mathcal{M}^d} \EX_{\vec{c} \sim P(k)}\!\left[\frac{v(A(\vec{c}))}{v(S^*(\vec{c}))} \right] \leq \frac{1}{2}\Big(1+\frac{1}{k}\Big).
\end{equation}

The first equality is due to the linearity of expectation. The second inequality follows by applying \Cref{lemma:yao} for the randomized mechanism $M$ and the  distribution $P(k)$, since all cost profiles in the support of $P(k)$ are members of $C(\hat{S}) = C(\{1,2\})$, i.e., the prediction $\{1,2\}$ is perfect for these instances. The last inequality follows from \Cref{lemma:ratio-P-ub}.

Using that $k = \ceil*{{2}/{\varepsilon}}$, we expand \eqref{eq:lb-proof-conclusion} to obtain:
\begin{equation*}
    \frac{1}{2-\varepsilon} \leq \frac{1}{2} \Big(1+\frac{1}{k}\Big)= 
    \frac{1}{2} +\frac{1}{2 \ceil*{{2}/{\varepsilon}}}\leq \frac{2+\varepsilon}{4}\,,
\end{equation*}
which is only possible when $\varepsilon = 0$. Since $\varepsilon >0$, this is a contradiction, completing the proof.
\end{proofof}
\section{Conclusions}
In this work we initiate the study of budget-feasible mechanisms in environments with predictions for submodular valuation functions (both monotone and non-monotone). For most of the paper (\Cref{sec:mono-submodular,sec:mono-sm-sampling,sec:nonmono-submodular}), we focus on the online setting with random (secretary) arrivals, for which we propose two families of tunable mechanisms that incorporate, as a prediction, the value of the optimal offline solution. Our results show that mechanisms augmented with this type of relatively weak prediction can substantially boost the performance of truthful and budget-feasible mechanisms for a wide class of procurement auctions. At the same time, through our lower bound in \Cref{sec:lower-bound}, we obtain strong evidence that a class of natural predictions on the output are largely ineffective for the offline version of the problem with additive valuations.

For the positive results of our paper, we analyzed the environment with predictions on the value of the optimal solution. This form of prediction, also suggested by prior work (see, e.g., \citet{balkanski23}), is relatively weak in an information-theoretic sense and can be seen as an easily obtainable aggregation of historical data (see also the discussions in \citet{christodoulou24} on output predictions). Nevertheless, we believe that studying online budget-feasible mechanism design with stronger or just different types of predictions (e.g., predictions on the set of agents that are part of the optimal solution) is an interesting direction for future work. 

A second direction we find interesting is the study of subadditive valuations in online environments, with and without prediction. Recently, in the offline setting, \citet{neogi_et_al:LIPIcs.ITCS.2024.84} showed a constructive approach to achieve a constant approximation for this class of valuations. Since an analogous result for the online setting with secretary arrivals remains elusive, it would be intriguing to understand whether predictions can be used to that end.

Finally, an open problem is to characterize the optimal
consistency--robustness Pareto frontier in the random-order model, complementing our
positive results with matching lower bounds. Even for the underlying
\emph{algorithmic}, non-strategic problems without predictions, unconditional lower
bounds in the random-order (secretary) model are scarce. For instance, the optimal
competitive ratio of the matroid secretary problem is a long-standing open question;
\citet{bahrani2021} have obtained impossibility results for certain classes of natural
algorithms. Moreover, to the best of our knowledge, no lower bounds are known for the
mechanism-design counterparts of these questions. We thus expect that obtaining such
lower bounds will require substantially new machinery, and we leave this as an
intriguing direction for future work.

\section*{Acknowledgements}
This work was partially supported by the project MIS 5154714 of the National Recovery and Resilience Plan Greece $2.0$ funded by the European Union under the NextGenerationEU Program. Moreover, this work was partially supported by the framework of the H.F.R.I call “Basic research Financing (Horizontal support of all Sciences)” under the National Recovery and Resilience Plan ``Greece 2.0'' funded by the European Union – NextGenerationEU (H.F.R.I. Project Number: 15877). Finally, it was supported by the Dutch Research Council (NWO) through the Gravitation Project NETWORKS, grant no.~024.002.003 and by the European Union under the EU Horizon 2020 Research and Innovation Program, Marie Sk\l{}odowska-Curie Grant Agreement, grant no.~101034253.

\Urlmuskip=1mu plus 1mu\relax
\bibliography{References}
\bibliographystyle{abbrvnat}
\newpage
\appendix


\section{Missing Material from  \Cref{sec:nonmono-submodular}}\label{app:sec5}

\lemtruthbfnonmono*
\begin{proof}
Regarding universal truthfulness, since the mechanism is a probability distribution on two different mechanisms, it suffices to argue about each of them separately. We will show that it holds for \Cref{mechanism:nonmono-sm-withoutpred} and the proof for \Cref{mechanism:nonmono-sm-pred} is similar. First, \Cref{mechanism:nonmono-sm-withoutpred} with probability $q$, runs Dynkin's, which is obviously truthful.   

Coming now to the second mechanism that is run with probability $1-q$,  fix an arrival sequence, as realized by the random arrival model, and fix also $\xi_1$, the number of agents that determine the set $N_1$.
Note that the agents have no control over the position in which they arrive. Hence, any agent $i \in N_1$ is always rejected by the mechanism, regardless of the cost she declares. 
Moreover, each agent $i \in N_2$ is offered a price $p_i$, which is independent of her declared cost, since agent $i$ has no control over her slot of arrival and thus cannot affect her marginal contribution.

For the sake of contradiction, assume that agent $j\in N_2$ can achieve a better outcome by misreporting a cost $b_j \neq c_j$, while all other agents keep the same reported costs. Suppose first that agent $j$ belongs to the winning set $S$, under the truthful profile.
Then by misreporting, she receives the same payment as long as her declared cost is below the threshold $p_j$ which is still the same as before, since the algorithm runs in exactly the same way up until the time that $j$ arrives. 
If the declared cost is above the threshold, she is rejected by the mechanism and receives a payment of 0. Hence, when $j \in S$, she cannot guarantee a better utility by misreporting her true cost. Suppose now that agent $j$ does not belong to the solution $S$, under truthful reporting. She certainly cannot benefit from reporting a higher cost than $c_j$, as that would still result in rejection by the mechanism.
Additionally, if she reports a lower cost than $c_j$ and she is added to the solution $S$, this means that the reason for rejection before was that the offered payment was less than her true cost $c_j$ (and not because of budget exhaustion). By reporting $b_j < c_j$, given that the payment offered by the mechanism remains the same, this results in negative utility for agent $j$. All of the above leads to the conclusion that no agent can benefit from misreporting her true cost.

For budget feasibility, note first that in the case the mechanism runs Dynkin's algorithm, the mechanism pays exactly $B$. Otherwise, the second inequality in the if-condition in line 13 guarantees that for the solution $S$ returned by the mechanism, $\sum_{i \in S} p_i \leq B$, which ensures the budget feasibility of the solution $S$.

Finally, through the first inequality of the if condition in line 13, individual rationality is ensured, since for any agent $i$ added in $S$, it holds that $p_i\geq c_i$. 
    
\end{proof}

\lemmechanismnonmonowithoutprediction*
\begin{proof}
    
This proof resembles a lot the proof of \Cref{lem:mechanism-without-prediction}, so again in order to avoid repetition we will present the proof briefly. We will use the same extra notation, that was introduced in the proof of \Cref{lem:mechanism-nonmono-with-prediction}. We discriminate between three cases, as we did in all the proofs of performance guarantees. 
\medskip

\noindent\textbf{Case 1:} \underline{Agent $i^*$ has significant value, i.e.: $v(i^*) > \frac{\beta}{z}\cdot \frac{1}{e}\cdot (\frac{1}{2}-\delta)\cdot v(S^*)$.}\\[4pt]
For this case,  with probability $q$ we run Dynkin's algorithm, which gives a  factor of $1/e$ over single-agent solutions. Therefore, independently of what our mechanism does with the remaining probability, the expected value for $v(S)$ will be at least:
\[E[v(S)] \geq \frac{q}{e}\cdot \frac{\beta}{z}\cdot \frac{1}{e}\cdot (\frac{1}{2}-\delta)\cdot v(S^*)=\frac{q\beta(1-2\delta)}{2e^2z}\cdot v(S^*).\]

\noindent For the case that $v(i^*) \leq \frac{\beta}{z}\cdot \frac{1}{e}\cdot (\frac{1}{2}-\delta)\cdot v(S^*)$, we can use Bernstein's inequality (\Cref{thm:bernstein}), as we did in \Cref{lem:Hoeffding} and get that with probability at least $1-2 \exp{\left(-\frac{\delta^2/ 2 }{(\frac{\beta}{z\cdot e}\cdot(\frac{1}{2}-\delta))\cdot(\frac{1}{4}+ \frac{\delta}{3})}\right)}$, it holds that $(\frac{1}{2}-\delta)\cdot v(S^*)\leq v(S_1^*)$, $(\frac{1}{2}-\delta)\cdot v(S^*)\leq v(S_2^*) $, where $\delta$ is a positive constant, that we will set to a certain value later.
We assume that we are in this highly probable event for the remainder of the proof.
\medskip

\noindent\textbf{Case 2:} \underline{Each agent $i\in S_2^*\setminus (S_1\cup S_2)$ rejects the posted price and $v(i^*) \leq \frac{\beta}{z}\cdot \frac{1}{e}\cdot (\frac{1}{2}-\delta)\cdot v(S^*)$.}\\[4pt]
 By using the same analysis as in the proof of \Cref{lem:mechanism-nonmono-with-prediction} case we obtain the same inequality:
    \begin{align*}
        v(S_2^*)\leq 2\cdot v(S_1)+2\cdot v(S_2)+2\sum\limits_{i\in S_2^*\setminus(S_1\cup S_2)} v(i|S_{j(i)}^{t_i})
    \end{align*}
    Since each agent in $S_2^*\setminus(S_1\cup S_2)$ rejects the posted price offered by the mechanism:
    \begin{gather*}
        v(S_2^*)\leq 2\cdot v(S_1) +2\cdot v(S_2)+ 2\sum\limits_{i\in S_2^*\setminus(S_1\cup S_2)} \frac{c_i\cdot \beta \cdot v(T_1)}{B}\implies\\ 2\cdot v(S_1)+2\cdot v(S_2)\geq v(S_2^*)-2\beta\cdot v(S^*),
    \end{gather*}
    where the last inequality due to the fact that $c(S_2^*)\leq B$ and $ v(T_1)\leq v(S^*)$. Therefore by dividing by 4, we get that $\EX[v(S)]\geq \frac{v(S_2^*)}{4}-\frac{\beta\cdot v(S^*)}{2}$, (since $\EX[v(S)]=\frac{v(S_1)}{2}+\frac{v(S_2)}{2}$).

    Combining the above inequality with the fact that $v(S_2^*)\geq(\frac{1}{2}-\delta)\cdot v(S^*) $: 
    \[\EX[v(S)]\geq \left( \frac{1}{8}-\frac{\delta}{4}-\frac{\beta}{2}\right) v(S^*)=\frac{1-2\delta-4\beta}{8}\cdot v(S^*),\] with probability at least $1-2 \exp{\left(-\frac{\delta^2/ 2 }{(\frac{\beta}{z\cdot e}\cdot(\frac{1}{2}-\delta))\cdot(\frac{1}{4}+ \frac{\delta}{3})}\right)}$. 
    \medskip
    
    \noindent\textbf{Case 3:} \underline{There is an agent $i\in S_2^*\setminus (S_1\cup S_2)$ such that $c_i\leq \bar{p}_i$ but $\bar{p}_i >B_{j(i)}$}\\ \underline{and $v(i^*) \leq \frac{\beta}{z}\cdot \frac{1}{e}\cdot (\frac{1}{2}-\delta)\cdot v(S^*)$.}\\[4pt]
    Let $i\in S_2^*\setminus (S_1\cup S_2)$ be the first agent that has this property.
    We start by showing that $B'_{j(i)}\leq\frac{B}{z}$, we have that:
    \begin{align*}
        B'_{j(i)}< \bar{p}_i =\frac{B}{t}(v(S_{j(i)}^{t_i}\cup \{i\})-v(S_{j(i)}^{t_i}))\leq \frac{B\cdot v(i)}{t}\leq \frac{B\cdot v(i^*)}{\beta \cdot v(T_1)} \leq \frac{B}{z}.
    \end{align*}
 The chain of inequalities uses submodularity ($v(i|S_{j(i)})\leq v(i)\leq v(i^*)$) and the assumption of Case~3 under the event $v(S_1^*)\geq (\frac{1}{2}-\delta)\cdot v(S^*)$, which gives $v(i^*)\leq \frac{\beta}{z}\cdot\frac{1}{e}\cdot(\frac{1}{2}-\delta)\cdot v(S^*)\leq \frac{\beta v(T_1)}{z}$.

\noindent So it holds that $B_{j(i)}'\leq \frac{B}{z}$, and thus, $\sum\limits_{k\in S_{j(i)}}p_k\geq B-\frac{B}{z}=\frac{(z-1)\cdot B}{z}$. Hence,
    \begin{align*}
        &\frac{z-1}{z}B\leq \sum\limits_{k\in S_{j(i)}}p_k=\frac{\sum\limits_{k\in S_{j(i)}} v(k|S_{j(i)}^{t_k}) \cdot B}{\beta \cdot v(T_1)}\\ & \Rightarrow v(S_{j(i)})\geq \frac{z-1}{z}\frac{\beta}{e} v(S_1^*)\\ &\Rightarrow \max\{v(S_1),v(S_2)\} \geq  \frac{z-1}{z}\cdot\frac{\beta}{e} v(S_1^*)
    \end{align*}
    Again, combining the last inequality with the fact that $v(S_1^*)\geq (\frac{1}{2}-\delta)\cdot v(S^*)$, with probability at least $1-2 \exp{\left(-\frac{\delta^2/ 2 }{(\frac{\beta}{z\cdot e}\cdot(\frac{1}{2}-\delta))\cdot(\frac{1}{4}+ \frac{\delta}{3})}\right)}$, we get that $\max\{v(S_1),v(S_2)\}\geq \frac{z-1}{z}\cdot\frac{\beta}{e}\cdot(\frac{1}{2}-\delta)\cdot v(S^*)$. So $\EX[X]\geq \frac{z-1}{2\cdot z}\cdot\frac{\beta}{e}\cdot(\frac{1}{2}-\delta)\cdot v(S^*)=\frac{(z-1)\beta(1-2\delta)}{4ez}\cdot v(S^*)$
    \\
Piecing everything together,

\begin{itemize}[leftmargin=*]
    \item For the case that $v(i^*) \leq \frac{\beta}{z}\cdot \frac{1}{e}\cdot (\frac{1}{2}-\delta)\cdot v(S^*)$, with probability $1-q$ we run the mechanism below line 4, which gives the factor \[(1-q) \cdot \tilde{p} \cdot  \min{\left( \frac{1-2\delta-4\beta}{8}, \frac{(z-1)\beta(1-2\delta)}{4ez} \right)},\] where $\tilde{p}=1-2 \exp{\left(-\frac{\delta^2/ 2 }{(\frac{\beta}{z\cdot e}\cdot(\frac{1}{2}-\delta))\cdot(\frac{1}{4}+ \frac{\delta}{3})}\right)}$. 
     \item For the case that $v(i^*) >\frac{\beta}{z}\cdot \frac{1}{e}\cdot (\frac{1}{2}-\delta)\cdot v(S^*)$ , with probability $q$ we run Dynkin's algorithm , which gives the factor
    \begin{equation}
        \frac{q\beta(1-2\delta)}{2e^2z}.\tag*{\qedhere}
    \end{equation}
\end{itemize} 

\end{proof}

\thmapproxnonmonosmtwo*
\begin{proof}

For this proof, we will consider the same three cases as we have done so far, for all the mechanisms.
\medskip

\noindent\textbf{Case 1:} \underline{Agent $i^*$ has significant value, i.e.: $v(i^*) >\frac{1}{z}\cdot ( a\cdot (1-\varepsilon)+ \frac{\beta}{e}\cdot(\frac{1}{k}-\delta))\cdot v(S^*)$.}\\[4pt]
For this case, we will utilize the fact that with probability $q$ we run Dynkin's algorithm, which gives a  factor of $1/e$ over single-agent solutions. Therefore, independently of what our mechanism does with the remaining probability, the expected value for $v(S)$ will be at least:
\[E[v(S)] \geq \frac{q}{e\cdot z}\cdot ( a\cdot (1-\varepsilon)+ \frac{\beta}{e}\cdot(\frac{1}{k}-\delta))\cdot v(S^*)=\frac{q \left(e k\,a(1 - \varepsilon) + \beta (1 - k \delta)\right)}{e^2 z k} \cdot v(S^*).\] 

\noindent Using a similar approach as in \Cref{lem:Hoeffding}, for the case that $v(i^*) \leq \frac{1}{z}\cdot ( a\cdot (1-\varepsilon)+  \frac{\beta}{e}\cdot(\frac{1}{k}-\delta))\cdot v(S^*)$, by applying Bernstein's inequality (\Cref{thm:bernstein}), we get that $(\frac{1}{k}-\delta)\cdot v(S^*)\leq v(S_1^*)$, $(\frac{k-1}{k}-\delta)\cdot v(S^*)\leq v(S_2^*) $, with probability greater or equal to $1-2\cdot \exp{\left(-\frac{\delta^2/ 2 }{(\frac{k-1}{k^2}+\frac{\delta}{3})\cdot\frac{1}{z}\cdot ( a\cdot (1-\varepsilon)+\frac{\beta}{e}\cdot (\frac{1}{k}-\delta))}\right)}$.
\medskip

\noindent\textbf{Case 2:} \underline{Each agent $i\in S_2^*\setminus (S_1\cup S_2)$ rejects the posted price and}\\ \underline{$v(i^*) \leq\frac{1}{z}\cdot ( a\cdot (1-\varepsilon)+ \frac{\beta}{e}\cdot(\frac{1}{k}-\delta))\cdot v(S^*)$.}\\[4pt]
In the same way as we did in the proof of \Cref{lem:mechanism-nonmono-with-prediction}, we obtain the two following inequalities: 
\begin{align*}
    v(S_2\cup S^*)\leq v(S_2)+v(S_1)+\sum\limits_{i\in S_2^*\setminus(S_1\cup S_2)} v(i|S_2^{t_i}), \\ v(S_1\cup S^*)\leq v(S_2)+v(S_1)+\sum\limits_{i\in S_2^*\setminus(S_1\cup S_2)} v(i|S_1^{t_i}).
\end{align*}
Adding them up and dividing by 4, we get:
\begin{align*}
\frac{v(S_1)}{2}+\frac{v(S_2)}{2}\geq \frac{v(S_2^*)}{4}-\frac{1}{2} \sum\limits_{i\in S_2^*\setminus(S_1\cup S_2)} v(i|S_{j(i)}^{t_i}),
\end{align*} 
which combined with the fact that each agent $i\in S_2^*\setminus(S_1\cup S_2)$ rejected the posted price offered by the mechanism and also that $\mathbb{E}[v(S)]=\frac{v(S_1)}{2}+\frac{v(S_2)}{2}$, we get that : 
\begin{align*}
\EX[v(S)]\geq  \frac{v(S_2^*)}{4}-\frac{ a\cdot(1-\varepsilon)\cdot v(S^*)+\beta\cdot v(T_1)}{2}\geq \frac{v(S_2^*)}{4}-\frac{(a\cdot(1-\varepsilon)+\beta)\cdot v(S^*)}{2}.
\end{align*}
Combining with the fact that $(\frac{k-1}{k}-\delta)\cdot v(S^*)\leq v(S_2^*) $, with probability greater or equal to $1-2\cdot \exp{\left(-\frac{\delta^2/ 2 }{(\frac{k-1}{k^2}+\frac{\delta}{3})\cdot\frac{1}{z}\cdot ( a\cdot (1-\varepsilon)+ \frac{\beta}{e}\cdot (\frac{1}{k}-\delta))}\right)}$, we obtain that for this case:
\[\EX[v(S)]\geq \left (\frac{1}{4}\cdot (\frac{k-1}{k}-\delta) -\frac{1}{2}\cdot ( a\cdot (1-\varepsilon)+ \beta)\right) \cdot v(S^*)=\frac{k-1-k\delta-2k(a(1-\varepsilon)+\beta)}{4k}\cdot v(S^*)\]
with probability at least $1-2 \exp{\left(-\frac{\delta^2/ 2 }{(\frac{k-1}{k^2}+\frac{\delta}{3})\cdot\frac{1}{z}\cdot ( a\cdot (1-\varepsilon)+ \frac{\beta}{e}\cdot (\frac{1}{k}-\delta))}\right)}$. \medskip

\noindent\textbf{Case 3:} \underline{There is an agent $i\in S_2^*\setminus (S_1\cup S_2)$ such that $c_i\leq \bar{p}_i$ but $\bar{p}_i >B_{j(i)}$ and also}\\ \underline{$v(i^*) \leq\frac{1}{z}\cdot ( a\cdot (1-\varepsilon)+ \frac{\beta}{e}\cdot(\frac{1}{k}-\delta))\cdot v(S^*)$.}\\[4pt]
Let $i\in S_2^*\setminus (S_1\cup S_2)$ be the first agent that has this property.
    First, we will show that $B'_{j(i)}\leq\frac{B}{z}$, we have that:
    \begin{align*}
        B'_{j(i)}<p_i=\frac{B}{t}(v(S_{j(i)}^{t_i}\cup \{i\})-v(S_{j(i)}^{t_i}))\leq \frac{B\cdot v(i)}{t}\leq \frac{B\cdot v(i^*)}{a\cdot \omega+\beta \cdot v(T_1)} \leq \frac{B}{z}, 
    \end{align*}
The chain of inequalities uses submodularity ($v(i|S_{j(i)})\leq v(i)\leq v(i^*)$) and the assumption of Case~3 under the event that $v(S_1^*)\geq (\frac{1}{2}-\delta)\cdot v(S^*)$, which gives $v(i^*)\leq \frac{1}{z}(a(1-\varepsilon)+\frac{\beta}{e}(\frac{1}{k}-\delta))v(S^*)\leq \frac{a\omega+\beta v(T_1)}{z}$.

\noindent Since $B'\leq \frac{B}{z}$, it means that  $\sum\limits_{k\in S_{j(i)}}p_k\geq B-\frac{B}{z}=\frac{(z-1)\cdot B}{z}$. So,
    \begin{align*}
        &\frac{z-1}{z}B\leq \sum\limits_{k\in S_{j(i)}}p_k=\frac{\sum\limits_{k\in S_{j(i)}} v(k|S_{j(i)}^{t_k}) \cdot B}{a\cdot \omega +\beta \cdot v(T_1)}\\ & \Rightarrow v(S_{j(i)})\geq \frac{z-1}{z}(a\cdot \omega+\frac{\beta}{e} \cdot v(S_1^*))\\ &\Rightarrow \max\{v(S_1),v(S_2)\} \geq  \frac{z-1}{z}\cdot(a\cdot \omega+\frac{\beta}{e}\cdot  v(S_1^*)).
    \end{align*}
    Again, combining the last inequality with the fact that $v(S_1^*)\geq (\frac{1}{k}-\delta)\cdot v(S^*)$, we get that $\max\{v(S_1),v(S_2)\} \allowbreak \geq \frac{z-1}{z}\cdot(a\cdot (1-\varepsilon)+\frac{\beta}{e}\cdot(\frac{1}{k}-\delta))\cdot v(S^*)$. So, 
    \[
      \EX[X]\geq \frac{z-1}{2\cdot z}\cdot(a\cdot (1-\varepsilon)+\frac{\beta}{e}\cdot(\frac{1}{k}-\delta))\cdot v(S^*)=\frac{(z - 1)\left(e k\,a(1 - \varepsilon) + \beta (1 - k\delta)\right)}{2 e z k}\cdot v(S^*),  
   \]
    with probability at least $1-2\cdot \exp{\left(-\frac{\delta^2/ 2 }{(\frac{k-1}{k^2}+\frac{\delta}{3})\cdot\frac{1}{z}\cdot ( a\cdot (1-\varepsilon)+ \frac{\beta}{e}\cdot (\frac{1}{k}-\delta))}\right)}$.
    \\
Piecing everything together,

\begin{itemize}[leftmargin=*]
    \item For the case that $v(i^*) \leq \frac{1}{z}\cdot (a\cdot (1-\varepsilon)+ \ \frac{\beta}{e}\cdot(\frac{1}{k}-\delta))\cdot v(S^*)$, with probability $1-q$ we run the mechanism below line 4, which gives the factor \[(1-q) \cdot \tilde{p}\cdot  \min{\left(\frac{k-1-k\delta-2k(a(1-\varepsilon)+\beta)}{4k}, \frac{(z - 1)\left(e k\,a(1 - \varepsilon) + \beta (1 - k\delta)\right)}{2 e z k} \right)},\] where $\tilde{p}=1-2 \exp{\left(-\frac{\delta^2/ 2 }{(\frac{k-1}{k^2}+\frac{\delta}{3})\cdot\frac{1}{z}\cdot ( a\cdot (1-\varepsilon)+\frac{\beta}{e}\cdot (\frac{1}{k}-\delta))}\right)}$. 
    
    \item For the case that $v(i^*) \ge\frac{1}{z}\cdot ( a\cdot (1-\varepsilon)+ \frac{\beta}{e}\cdot(\frac{1}{k}-\delta))\cdot v(S^*)$ , with probability $q$ we run Dynkin's algorithm , which gives the factor
    \begin{equation}
        \frac{q \left(e k\,a(1 - \varepsilon) + \beta (1 - k \delta)\right)}{e^2 z k}. \tag*{\qedhere}
    \end{equation}
\end{itemize} 

\end{proof}


\section{Missing Material from  \Cref{sec:lower-bound}}\label{app:sec6}
\lemmayao*
\begin{proof}
  We use $\supp(D)$ to denote the support of a discrete probability distribution $D$; we extend this notation to randomized mechanisms, since each such mechanism defines a discrete probability distribution over deterministic mechanisms.
  Also, let $\Bar{\vec{c}}$ be a random variable on cost profiles. For every mechanism $(A^r, \vec{p}^r) \in \mathcal{M}$ and every probability distribution $P$ over canonical instances with $\supp(P) \subseteq C$, the following holds:
\begin{align*}
    \min_{\vec{c} \in C}\EX\left[\frac{v\left(A^r(\vec{c}) \right)}{v\left(S^*(\vec{c}) \right)} \right] 
    &\leq \min_{\vec{c} \in \supp(P)}\EX\left[\frac{v\left(A^r(\vec{c}) \right)}{v\left(S^*(\vec{c}) \right)} \right] \\
    &= \min_{\vec{c} \in \supp(P)} \sum_{(A, \vec{p}) \in \supp((A^r, \vec{p}^r))}\Pr[A^r=A] \cdot \frac{v\left(A(\vec{c}) \right)}{v\left(S^*(\vec{c}) \right)} \\
    &\leq \sum_{\vec{c} \in \supp(P)}\Pr[\Bar{\vec{c}}= \vec{c}] \cdot \sum_{(A, \vec{p}) \in \supp((A^r, \vec{p}^r))}\Pr[A^r=A] \cdot \frac{v\left(A(\vec{c}) \right)}{v\left(S^*(\vec{c}) \right)} \\
    &= \sum_{(A, \vec{p}) \in \supp((A^r, \vec{p}}^r))\Pr[A^r=A] \cdot \sum_{\vec{c} \in \supp(P)}\Pr[\Bar{\vec{c}}= \vec{c}] \cdot \frac{v\left(A(\vec{c}) \right)}{v\left(S^*(\vec{c}) \right)} \\
    &= \sum_{(A, \vec{p}) \in \supp((A^r, \vec{p}^r))}\Pr[A^r=A] \cdot \EX_{\Bar{\vec{c}} \sim P}\left[\frac{v\left(A(\Bar{\vec{c}}) \right)}{v\left(S^*(\Bar{\vec{c}}) \right)}\right] \\
    &\leq \max_{(A, \vec{p}) \in \supp((A^r, \vec{p}^r))}\EX_{\Bar{\vec{c}} \sim P}\left[\frac{v\left(A(\Bar{\vec{c}}) \right)}{v\left(S^*(\Bar{\vec{c}}) \right)}\right] \\
    &\leq \max_{(A, \vec{p}) \in \mathcal{M}^d}\EX_{\Bar{\vec{c}} \sim P}\left[\frac{v\left(A(\Bar{\vec{c}}) \right)}{v\left(S^*(\Bar{\vec{c}}) \right)}\right].
\end{align*}
The second and third inequalities follow from the standard fact that, for a set of $k$ numbers $\{x_1, \dots, x_k\}$ and non-negative weights $w_1, \dots, w_k$ with $\sum_{i=1}^k w_i = 1$, it holds that $
\min_{i \in [k]} x_i \leq \sum_{i=1}^k w_i x_i \leq \max_{i \in [k]} x_i$.
Finally, the last inequality follows from the fact that, for any mechanism $(A^r, \vec{p}^r) \in \mathcal{M}$, it holds that $\supp(A^r, \vec{p}^r) \subseteq \mathcal{M}^d$. 
\end{proof}

\end{document}